%% file: HH.tex
\documentclass[11pt,letterpaper]{article}
\usepackage{amsmath,amssymb}
\usepackage{amsthm,amsfonts}
\usepackage{enumerate}
\usepackage{geometry}
\usepackage{algorithm,algpseudocode}
\usepackage{url}
\usepackage{etoolbox}

\input{macros.tex}
\newtoggle{shortversion}
\togglefalse{shortversion} 

\iftoggle{shortversion}{
  \geometry{margin=1in}
}{}

\makeatletter
\newcommand{\manuallabel}[2]{\def\@currentlabel{#2}\label{#1}}
\makeatother

\title{Beating CountSketch for Heavy Hitters in Insertion Streams}
\author{Vladimir Braverman\thanks{“This material is based upon work supported in part by the
National Science Foundation under Grant No. 1447639, by the Google
Faculty Award and by DARPA grant N660001-1-2-4014. Its contents are
solely the responsibility of the authors and do not represent the
official view of DARPA or the Department of Defense.} \and Stephen R.\ Chestnut \thanks{“This material is based upon work supported in part by the
National Science Foundation under Grant No. 1447639}
\and Nikita Ivkin \thanks{“This material is based upon work supported in part by the
National Science Foundation under Grant No. 1447639 
and by DARPA grant N660001-1-2-4014. Its contents are
solely the responsibility of the authors and do not represent the
official view of DARPA or the Department of Defense.} \and David P.\ Woodruff}

\begin{document}
\maketitle
\input{abstract.tex}

\thispagestyle{empty}
\newpage
\setcounter{page}{1}

\input{MainTheorems}

\section{Introduction}
\input{Sec-Intro}

\subsection{Preliminaries}\label{sec:prelim}
\input{Sec-Notations}

\section{$\ell_2$ heavy hitters algorithm}\label{sec:main}
\input{Sec-Algorithms}

\section{Chaining Inequality}\label{sec:chaining}
\input{Sec-Chaining}

\iftoggle{shortversion}{}{
\section{Reduced randomness}\label{sec:randomness}
\input{Sec-Derandomization}

\section{$F_2$ at all points}\label{sec:f2all}
\input{Sec-F2always}

}
\section*{Acknowledgements}
The authors would like to thank Raghu Meka for suggesting the JL strategy to reduce the number of random bits needed for the Bernoulli processes.
We would also like to thank several anonymous reviewers for their careful reading of an earlier version of this paper and the valuable suggestions that they made.

\bibliographystyle{plain}
\bibliography{main}

\end{document}

%% file: macros.tex
\newtheorem{theorem}{Theorem}

\newtheorem{definition}[theorem]{Definition}
\newtheorem{lemma}[theorem]{Lemma}

\newtheorem{corollary}[theorem]{Corollary}

\newtheorem{remk}[theorem]{Remark}
\newtheorem{exmp}[theorem]{Example}

\newcommand{\R}{{\mathbb R}}

\newcommand{\poly}{{\mathrm{poly}}}
\newcommand{\polylog}{{\mathrm{polylog}}}
\newcommand{\loglog}{{\mathop{\mathrm{loglog}}}}

\newcommand{\eps}{\varepsilon}

\DeclareMathOperator{\argmax}{argmax}

\newcommand{\calD}{\mathcal{D}}
\newcommand{\E}{\mathbb{E}}
\newcommand{\calH}{\mathcal{H}}
\newcommand{\calN}{\mathcal{N}}
\newcommand{\tN}{\widetilde{N}}

\newcommand{\tX}{\widetilde{X}}

\newcommand{\iidsim}{\stackrel{\text{iid}}{\sim}}
\newcommand{\fwsim}{\stackrel{\text{4-w}}{\sim}}
\newcommand{\pwsim}{\stackrel{\text{p.w.}}{\sim}}
\newcommand{\Bernoulli}{\text{Bernoulli}}
\newcommand{\Rademacher}{\text{Rademacher}}

\DeclareMathOperator{\Var}{Var}

%% file: abstract.tex
\begin{abstract}
Given a stream $p_1, \ldots, p_m$ of items from a universe
$\mathcal{U}$, which, without loss of generality we identify with
the set of integers $\{1, 2, \ldots, n\}$, 
we consider the problem of returning all
$\ell_2$-heavy hitters, i.e., those items $j$ for which
$f_j \geq \eps \sqrt{F_2}$, where $f_j$ is the number of occurrences
of item $j$ in the stream, and $F_2 = \sum_{i \in [n]} f_i^2$.
Such a guarantee is considerably stronger than the 
$\ell_1$-guarantee, which finds those $j$ for which 
$f_j \geq \eps m$. In
2002, Charikar, Chen, and Farach-Colton suggested the {\sf CountSketch}
data structure, which finds all such $j$ using $\Theta(\log^2 n)$ bits of space
(for constant $\eps > 0$). The only
known lower bound is $\Omega(\log n)$ bits of space, which comes from 
the need to specify the identities of the items found.

In this paper we show it is possible to achieve $O(\log n \log \log n)$ bits of space for this problem.
Our techniques, based on Gaussian processes, lead to a number of other new results for data streams, including
\begin{enumerate}
\item The first algorithm
for estimating $F_2$ simultaneously at all points in a stream
using only $O(\log n\log\log n)$ bits of space, improving 
a natural union bound and the algorithm of Huang, Tai, and Yi (2014).
\item A way to estimate the $\ell_{\infty}$ norm of a stream up to
additive error $\eps \sqrt{F_2}$ with $O(\log n\loglog n)$ bits of space, resolving
Open Question 3 from the IITK 2006 list for insertion only streams.
\end{enumerate}
\end{abstract}

%% file: MainTheorems.tex
\newcommand{\heavyhitterstheorem}{For any $\epsilon>0$, there is a $1$-pass algorithm in the insertion-only model that, with probability at least $2/3$, finds all those indices $i \in [n]$ for which $f_i \geq \epsilon \sqrt{F_2}$, and reports no indices $i \in [n]$ for which $f_i \leq \frac{\epsilon}{2} \sqrt{F_2}$. The space complexity is  $O(\frac{1}{\epsilon^2}\log\frac{1}{\epsilon}\log n\log\log n)$~bits.}

\newcommand{\ftwoalwaystheorem}{
For any $\epsilon>0$, there is a $1$-pass algorithm in the insertion-only model that, with probability at least $2/3$, outputs a $(1\pm\epsilon)$-approximation of $F_2$ at all points in the stream.
  The algorithm uses $O(\frac{1}{\epsilon^2}\log n(\log{\frac{1}{\epsilon}}+\log\log n))$ bits of space.}

%% file: Sec-Intro.tex
There are numerous applications of data streams, for which the elements $p_i$ may be numbers, points, edges in a graph, and so on. Examples include internet search logs, network traffic, sensor networks, and scientific data streams (such as in astronomy, genomics, physical simulations, etc.). 
The sheer size of the dataset often imposes very stringent requirements on an algorithm's resources.
In many cases only a single pass over the data is feasible, such as in network applications, since if the data on a network is not physically stored somewhere, it may be impossible to make a second pass over it. 
There are multiple surveys and tutorials in the algorithms, database, and networking communities on the recent activity in this area; we refer the reader to \cite{m05,bbdmw02} for more details and motivations underlying this area. 

Finding heavy hitters, also known as the top-$k$, most popular items, elephants, or iceberg queries, is arguably one of the most fundamental problems in data streams. 
It has applications in flow identification at IP routers \cite{ev03}, iceberg queries \cite{fsgmu98}, iceberg datacubes \cite{br99,hpdw01}, association rules, and frequent itemsets \cite{as94,son95,toi96,hid99,hpy00}.

Formally, we are given a stream $p_1, \ldots, p_m$ of items from a universe $\mathcal{U}$, which, without loss of generality we identify with the set $\{1, 2, \ldots, n\}$. We make the common assumption that $\log m = O(\log n)$, though our results generalize naturally to any $m$ and $n$. Let $f_i$ denote the frequency, that is, the number of occurrences, of item $i$. We would like to find those items $i$ for which $f_i$ is large, i.e., the ``heavy hitters''. 
In this paper we will consider algorithms that are allowed one pass over the stream and must use as little space (memory) in bits as possible, to maintain a succinct summary (``sketch'') so that after processing the stream, we can output the identities of all of the heavy hitters from the summary with large probability.

There are various notions of what it means for $f_i$ to be large. One such notion is that we should return all indices $i \in [n]$ for which $f_i \geq \epsilon m$ for a parameter $\epsilon \in (0,1)$, and no index $i$ for which $f_i \leq (\epsilon - \phi)m$, for a parameter $\phi$.
It is typically assumed that $\phi \geq c \epsilon$ for an absolute constant $c > 0$, and we will do so in this paper. 
This notion has been extensively studied, so much so, that the same streaming algorithm for it was re-invented multiple times. The first algorithm was given by Misra and Gries~\cite{mg82}, who achieved $O((\log n)/\epsilon)$ bits of space. 
The algorithm was rediscovered by Demaine et al.~\cite{dlm02}, and then later rediscovered by Karp et al.~\cite{ksp03}. 
Cormode and Muthukrishan~\cite{cm03} state that ``papers on frequent items are a frequent item!''. While these algorithms are deterministic,
there are also several randomized algorithms, including the Count-Min sketch~\cite{cm05}, sticky sampling~\cite{mm12}, lossy counting~\cite{mm12}, sample and hold~\cite{ev03}, multi-stage bloom filters~\cite{cfm09}, sketch-guided sampling~\cite{kx06}, and CountSketch~\cite{ccf04}. 
A useful (slightly suboptimal) intuition is that one can sample $O((\log 1/\epsilon)/\epsilon)$ random stream positions, remember the identities of these positions, and then maintain the counts of these items. 
By coupon collector arguments, all heavy hitters will be found this way, and one can filter out the spurious ones (those with $f_i \leq (\epsilon-\phi)m$). 

One of these techniques, CountSketch~\cite{ccf04}, refined in~\cite{tz12}, gives a way of finding the {\it $\ell_2$-heavy hitters} of a stream. 
Those are the items for which $f_i^2 \geq \eps^2 F_2$. 
Notice that this guarantee is significantly stronger than the aforementioned guarantee that $f_i \geq \epsilon m$, which we will call the $\ell_1$-guarantee.  
Indeed, if $f_i \geq \epsilon m$, then $f_i^2 \geq \epsilon^2 m^2 \geq \epsilon^2 F_2$. 
So, an algorithm for finding the $\ell_2$-heavy hitters will find all items satisfying the $\ell_1$-guarantee. 
On the other hand, given a stream of $n$ distinct items in which $f_i = \sqrt{n}$ for an $i \in [n]$, yet $f_j = 1$ for all $j \neq i$, an algorithm satisfying the $\ell_2$-heavy hitters guarantee will identify item $i$ with constant $\epsilon$, but an algorithm which only has the $\ell_1$-guarantee would need to set $\epsilon = 1/\sqrt{n}$, using $\Omega(\sqrt{n} \log n)$ bits of space. 
In fact, $\ell_2$-heavy hitters are in some sense the best one can hope for with a small amount of space in a data stream, as it is known for $p > 2$ that finding those $i$ for which $f_i^p \geq \eps^p F_p$ requires $n^{1-2/p}$ bits of space~\cite{bjks04,cks03}. 

The CountSketch has broad applications in compressed sensing~\cite{glps10,p11,mp14} and numerical linear algebra~\cite{cw13,mm13,nn13,bn13}, and are often used as a subroutine in other data stream algorithms, such as  $\ell_p$-sampling~\cite{mw10,ako11,jst11}, cascaded aggregates~\cite{jw09}, and frequency moments~\cite{iw05,bgks06}. 

Given the strong guarantees and many applications of $\ell_2$-heavy hitter algorithms, it is natural to ask what the best space complexity for them is. 
Both the original algorithm of~\cite{ccf04} and the followup of~\cite{tz12} achieve $\Theta(\log^2 n)$ bits of space for constant values of $\epsilon$. 
On the other hand, the only known lower bound is $\Omega(\log n)$ bits, which is needed just to identify the heavy hitter. 
Despite the success we have had in obtaining space-optimal streaming algorithms for estimating moments and $p$-norms, this has remained a glaringly open problem. 
It is known that if one allows deletions in the stream, in addition to insertions, then $\Theta(\log^2 n)$ bits of space is optimal~\cite{bipw11,jst11}. 
However, in many cases we just have a stream of insertions, such as in the model studied in the seminal paper of Alon, Matias, and Szegedy~\cite{ams99}.

\subsection{Our Contributions}
The main result of this paper is the near resolution of the open question above. We show:

\begin{theorem}[$\ell_2$-Heavy Hitters]\label{thm:main}
\heavyhitterstheorem
\end{theorem}
The intuition of the proof is as follows. 
Suppose there is a single $\ell_2$-heavy hitter $H$, $\eps > 0$ is a constant, and we are trying to find the identity of $H$. 
Suppose further we could identify a substream $S'$ where $H$ is very heavy, specifically we want that the frequencies in the substream satisfy $\frac{f_H^2}{\poly(\log n)}\geq \sum_{j \in S',j \neq H} f_j^2$. 
Suppose also that we could find certain $R=O(\log{n})$ ``breakpoints'' in the stream corresponding to jumps in the value of $f_H$, that is, we knew a sequence  $p_{q_1} < p_{q_2} < \cdots < p_{q_{R}}$ which corresponds to positions in the stream for which $f_H$ increases by a multiplicative factor of $\left (1 + \frac{1}{\Theta(R)} \right)$. 

Given all of these assumptions, in between breakpoints we could partition the universe randomly into two pieces and run an $F_2$-estimate \cite{ams99} (AMS sketch) on each piece.
Since $f_H^2$ is more than a $\poly(\log n)$ factor times $\sum_{j \in S', j \neq H} f_j^2$, while in between each breakpoint the squared frequency of $H$ is $\Omega \left(\frac{f_H^2}{\log n} \right)$, it follows that $H$ contributes a constant fraction of the $F_2$-value in between consecutive breakpoints, and so, upon choosing the constants appropriately, the larger of the magnitudes of the two AMS sketches will identify a bit of information about $H$, with probability say 90\%. 
This is our algorithm Sieve. 
Since we have $\Theta(\log n)$ breakpoints, in total we will learn all $\log n$ bits of  information needed to identify $H$. 
One persepective on this algorithm is that it is a sequential implementation of the multiple repetitions of CountSketch, namely, we split the stream at the breakpoints and perform one ``repetition'' on each piece while discarding all but the single bit of information we learn about $H$ in between breakpoints. 

However, it is not at all clear how to (1) identify $S'$ and (2) find the breakpoints. 
For this, we resort to the theory of Gaussian and Bernoulli processes. 
Throughout the stream we can maintain a sum of the form $X_t = \sum_{i=1}^n f^{(t)}_i Z_i$, where the $Z_i$ are independent Normal$(0,1)$ or Rademacher random variables.
Either distribution is workable.
One might think as one walks through a stream of length $\poly(n)$, there will be times for which this sum is much larger than $\sqrt{F_2}$; indeed, the latter is the standard deviation and a na\"{i}ve union bound, if tight, would imply positions in the stream for which $|X_t|$ is as large as $\sqrt{F_2 \log n}$.
It turns out that this cannot happen! 
Using a generic chaining bound developed by Fernique and Talagrand \cite{t96}, we can prove that
there exists a universal constant $C'$ such that
\[\E\sup_t|X_t| \leq C'\sqrt{F_2}.\]
We call this the Chaining Inequality.

We now randomly partition the universe into $O(\frac{1}{\epsilon^2})$ ``parts'', and run our algorithm independently on each part. 
This ensures that, for a large constant $C$, $H$ is $C$-heavy, meaning, $f_H^2 \geq C(F_2-f_H^2)$, where here we abuse notation and use $F_2$ to denote the moment of the part containing $H$. 
We run the following two stage algorithm independently on each part.
The first stage, called Amplifier, conists of $L=O(\log \log n)$ independent and concurrent repetitions of the following: randomly split the set of items into two buckets and maintain a two Bernoulli processes, one for the updates in each bucket.
By the Chaining Inequality, a Markov bound, and a union bound, the total $F_2$ contribution, excluding that of $H$, in each piece in each repetition {\it at all times in the stream} will be $O(\sqrt{F_2-f_H^2})$.
Since $H$ is sufficiently heavy, this means after some time $t^*$, its piece will be larger in magnitude in most, say 90\%, of the $L$ repetitions. 
Furthermore, $H$ will be among only $n/2^{\Omega(L)}=n/\poly\log n$ items with this property. 
At this point we can restrict our attention to a substream containing only those items.

The substream has the property that its $F_2$ value, not counting $H$, will be a factor $\frac{1}{\log^2 n}$ times the $F_2$ value of the original stream, making $H$ $\Omega(\log^2 n)$-heavy.
Finally, to find the breakpoints, our algorithm Timer maintains a Bernoulli process on the substream, and every time the Bernoulli sum increases by a multiplicative $\left(1+\frac{1}{\theta(R)} \right)$ factor, creates a new breakpoint.
By the Chaining Inequality applied in each consecutive interval of breakpoints, the $F_2$ of all items other than $H$ in the interval is at most $O(\log n)$ larger than its expectation; while the squared frequency of $H$ on the interval is at least $\frac{f_H^2}{\log n}$. 
Since $H$ is $\Omega(\log^2 n)$-heavy, this makes $f_H^2$ to be the dominant fraction of $F_2$ on the interval. 

One issue with the techniques above is they assume a large number of random bits can be stored. 
A standard way of derandomizing this, due to Indyk \cite{i06} and based on Nisan's pseudorandom generator PRG \cite{n92}, would increase the space complexity by a $\log n$ factor, which is exactly what we are trying to avoid. 
Besides, it is not clear we can even apply Indyk's method since our algorithm decides at certain points in the stream to create new $F_2$-sketches based on the past, whereas Indyk's derandomization relies on maintaining a certain type of linear sum in the stream, so that reordering of the stream does not change the output distribution.
A first observation is that the only places we need more than limited independence are in  maintaining a collection of $O(\log{n})$ hash functions and the stochastic process $\sum_{i=1}^n f_i Z_i$ throughout the stream. 
The former can, in fact, be derandomized along the lines of Indyk's method~\cite{i06}.

In order to reduce the randomness needed for the stochastic process we use a Johnson-Lindenstrauss transformation to reduce the number of Rademacher (or Gaussian) random variables needed.
The idea is to reduce the frequency vector to $O(\log n)$ dimensions with JL and run the Bernoulli process in this smaller dimensional space.
\iftoggle{shortversion}{}{The Bernoulli process becomes $\sum_{i=1}^{O(\log n)}Z_i (T f)_i$, where $T$ is the JL matrix.}
The same technique is used by Meka for approximating the supremum of a Gaussian process~\cite{meka2012ptas}.
It works because the Euclidean length of the frequency vector describes the variance and covariances of the process, hence the transformed process has roughly the same covariance structure as the original process.
\iftoggle{shortversion}{}{An alternative perspective on this approach is that we use the JL transformation in reverse, as a pseudorandom generator that expands $O(\log n)$ random bits into $O(n)$ random variables which fool our algorithm using the Bernoulli process.}

\iftoggle{shortversion}{
In the full version
}{
In Section~\ref{sec:f2all}
}
we also use our techniques to prove the following.
\begin{theorem}[$F_2$ at all points]\label{thm: f2always}
\ftwoalwaystheorem
\end{theorem}

\paragraph{Outline}
In Section~\ref{sec:prelim}, we give preliminaries and define our notation. 
In Section~\ref{sec:main} we prove Theorem~\ref{thm:main}.
The proof of the Chaining Inequality for Gaussian and Bernoulli processes, the central tool used in Section~\ref{sec:main}, appears in Section~\ref{sec:chaining}.
\iftoggle{shortversion}{
In the full version we give complete proofs, including details about how to implement the algorithm with a reduced number of random bits, and we prove Theorem~\ref{thm: f2always}. 
}{
In Section \ref{sec:randomness} we give details about how to implement the algorithm with a reduced number of random bits. 
In Section \ref{sec:f2all} we prove Theorem~\ref{thm: f2always}. 
}

%% file: Sec-Notations.tex
\begin{table}
\begin{tabular}{ |l|l|l||l|l|l|  }
\hline
$L$ & amplifier size & $O(\log\log{n})$ &
$\tau$ & round expansion & $100(R+1)$\\\hline
$\delta$ & small constant & $\Omega(1)$ &
$S^{t_1:t_2}$ & interval of the stream & $(p_{t_1+1},\ldots,p_{t_2})$\\\hline
$H$ & heavy hitter id & $\in[n]$& 
$e_j$ & $j$th unit vector & \\\hline
$T$ & JL transformation & $\in\R^{k\times n}$ &
$f_H^{(k)}$ & frequency on $S^{0:k}$ & \\\hline
$m$ & stream length & $\poly(n)$ &
$f^{(k_1:k_2)}$	& frequency on $S^{k_1:k_2}$& $f^{(k_2)} - f^{(k_1)}$\\\hline
$n$ & domain size &  &
$R$ & $\#$ of Sieve rounds & $O(\log{n})$ \\ \hline
$k$ & JL dimension & $O(\log n)$ &
$C'$ & Chaining Ineq. const. & $O(1)$\\\hline
$d$ & dim. of Bern.\ proc.\ & $O(\log \delta^{-1})$ &
$C$ & large const. & $\geq d^{\frac{3}{2}}C'/\delta$\\\hline
\end{tabular}
\caption{Notation and parameters used throughout the paper.}\label{tab: notation}
\end{table}

Given a stream $S=(p_1,p_2,\ldots,p_m)$, with $p_i\in[n]$ for all $i$, we define the frequency vector at time $0\leq t\leq m$ to be the vector $f^{(t)}$ with coordinates $f_j^{(t)} := \#\{t'\leq t\mid p_{t'}=j\}$.
When $t=m$ we simply write $f:=f^{(m)}$.
Given two times $t_1\leq t_2$ we use $f^{(t_1:t_2)}$ for the vector $f^{(t_2)}-f^{(t_1)}$.
Notice that all of these vectors are nonnegative because $S$ has no deletions.
An item $H\in[n]$ is said to be an $\alpha$-heavy hitter, for $\alpha>0$, if $f_H^2\geq\alpha\sum_{j\neq H}f_j^2$.
\iftoggle{shortversion}{
  }{
  The goal of our main algorithm, CountSieve, is to identify a single $\alpha$-heavy hitter for $\alpha$ a large constant.
We will assume $\log{m}=O(\log{n})$, although our methods apply even when this is not true.
It will be occasionally helpful to assume that $n$ is sufficiently large.
This is without loss of generality since in the case $n=O(1)$ the problem can be solved exactly in $O(\log{m})$ bits.

A streaming algorithm is allowed to read one item at a time from the stream in the order given.
The algorithm is also given access to a stream of random bits, it must pay to store any bits that it accesses more than once, and it is only required to be correct with constant probability strictly greater than $1/2$.
Note that by repeating such an algorithm $k$ times and taking a majority vote, one can improve the success probability to $1-2^{-\Omega(k)}$.
\iftoggle{shortversion}{}{We measure the storage used by the algorithm on the worst case stream, i.e.\ worst case item frequencies and order, with the worst case outcome of its random bits.}
}

The AMS sketch~\cite{ams99} is a linear sketch for estimating $F_2$.
The sketch contains $O(\epsilon^{-2}\log{\delta^{-1}})$ independent sums of the form $\sum_{j=1}^n S_jf_j$, where $S_1,S_2,\ldots,S_n$ are four-wise independent Rademacher random variables.
By averaging and taking medians it achieves a $(1\pm\epsilon)$-approximation to $F_2$ with probability at least $(1-\delta)$.

A \emph{Gaussian process} is a stochastic process $(X_t)_{t\in T}$ such that every finite subcollection $(X_t)_{t\in T'}$, for $T'\subseteq T$, has a multivariate Gaussian distribution.
When $T$ is finite (as in this paper), every Gaussian process can be expressed as a linear transformation of a multivariate Gaussian vector with mean 0 and covariance $I$.
Similarly, a \emph{Bernoulli process} $(X_t)_{t\in T}$, $T$ finite, is a stochastic process defined as a linear tranformation of a vector of i.i.d.\ Rademacher (i.e.\ uniform $\pm1$) random variables.
Underpinning our results is an analysis of the Gaussian process $X_t = \sum_{j\in[n]} Z_j f_j^{(t)}$, for $t=0,\ldots,m$, where $Z_1,\ldots,Z_n\iidsim\calN(0,1)$ are independent standard Normal random variables.
The Bernoulli analogue to our Gaussian process replaces the distribution of the random vector $Z$ as $Z_1,\ldots,Z_n\iidsim\Rademacher$.
\iftoggle{shortversion}{}{Properties of the Normal distribution make it easier for us to analyze the Gaussian process rather than its Bernoulli cousin.
On the other hand, we find Bernoulli processes more desirable for computational tasks.}
Existing tools, which we discuss further in Section~\ref{sec:chaining} and \iftoggle{shortversion}{
  the full version}{
  Section~\ref{sec:randomness}}, allow us to transfer the needed properties of a Gaussian process to its Bernoulli analogue.

A $k\times n$ matrix $T$ is a \emph{$(1\pm\gamma)$-embedding} of a set of vectors $X\subseteq\R^n$ if
\iftoggle{shortversion}{
$(1-\gamma)\|x-y\|_2\leq \|Tx-Ty\|_2\leq (1+\gamma)\|x-y\|_2$,
  }{
  \[(1-\gamma)\|x-y\|_2\leq \|Tx-Ty\|_2\leq (1+\gamma)\|x-y\|_2,\]
}
for all $x,y\in X\cup\{0\}$.
We also call such a linear transformation a \emph{JL Transformation}.
It is well-known that taking the entries of the matrix $T$ to be i.i.d.\ Normal random variables with mean 0 and variance $1/k$ produces a JL transformation with high probability.
Many other randomized and deterministic constructions exist, we will use the recent construction of Kane, Meka, and Nelson~\cite{kane2011almost}.

The development and analysis of our algorithm relies on several parameters, some of which have already been introduced.
Table~\ref{tab: notation} lists those along with the rest of the parameters and some other notation for reference.
In particular, the values $C$, $d$, $\delta$, and $\gamma$ are constants that we will choose in order to satisfy several inequalities.
We will choose $\delta$ and $\gamma$ to be small, say $1/200$, and $d=O(\log 1/\delta)$.
$C$ and $C'$ are sufficiently large constants, in particular $C\geq dC'/\delta$.

%% file: Sec-Algorithms.tex
\iftoggle{shortversion}{}{
This section describes the algorithm CountSieve, which solves the heavy-hitter problem for the case of a single heavy hitter, i.e.\ top-1,  in $O(\log{n}\log\log n)$ bits of space and proves Theorem~\ref{thm:main}.
By definition, the number of $\epsilon$-heavy hitters is at most $1+1/\epsilon$, so, upon hashing the universe into $O(1/\epsilon^2)$ parts, the problem of finding all $\epsilon$-heavy hitters reduces to finding a single heavy hitter in each part.
When $\epsilon=\Omega(1)$, using this reduction incurs only a constant factor increase in space over the single heavy hitter problem.

Suppose the stream has only a single heavy hitter $H\in[n]$.
Sequentially, over the course of reading the stream, CountSieve will hash the stream into two separate substreams for $O(\log{n})$ repetitions, and in each repetition it will try to determine which of the two substreams has the heavy hitter using the AMS Sketch.
With high probability, $H$ has a unique sequence of hashes, so if we correctly identify the stream containing $H$ every time then we can correctly identify $H$. This holds even if we only correctly identify the stream containing $H$ a large constant fraction of the repetitions.
CountSketch accomplishes this by performing the $O(\log{n})$ rounds of hashing in parallel, with $\Omega(\log^2n)$ bits of storage.
One of our innovations is to implement this scheme \emph{sequentially} by specifying intervals of updates, which we call \emph{rounds}, during each of which we run the two AMS Sketches.
In total there could be as many as  $\Theta(\log^2{n})$ of these rounds, but we will discard all except the last $R=O(\log n)$ of them.
}

Algorithm~\ref{algo: BP} is a simplified version of the Bernoulli process used by CountSieve.
It has all of the properties we need for correctness of the algorithm, but it requires too many random bits. 
Chief among these properties is the control on the supremum of the process.
\begin{algorithm}[h!]
  \begin{algorithmic}
    \Procedure{BP}{Stream $S$}
    \State Sample $Z_1,\ldots,Z_n\iidsim\Rademacher$
    \State \Return {$\langle Z, f^{(t)}\rangle$ at each time $t$}
    \EndProcedure
  \end{algorithmic}
  \caption{One Bernoulli process.}\label{algo: BP}
\end{algorithm}
The Chaining Inequality gives us a uniform bound on the maximum value of the BP process in terms of the standard deviation of the last value.
This property is formalized by the definition of a tame process.

\begin{definition}\label{def: tame}
  Let $f^{(t)}\in\R^n$, for $t\in[m]$, and let $T:\R^n\to\R^k$ be a matrix. 
  Let $Z$ be a $d\times k$ matrix of i.i.d.\ Rademacher random variables.
  A $d$-dimensional Bernoulli process $y_t = d^{-\frac{1}{2}}ZTf^{(t)}$, for $t\in[m]$, is \emph{tame} if, with probability at least $1-\delta$, 
\begin{equation}\label{eq: tame definition}
\|y_t\|_2 \leq C\sqrt{\sum_{j=1}^n f_j^2},\quad\text{for all }t\in[m].
\end{equation}
\end{definition}

The definition anticipates our need for dimension reduction in order to reduce the number of random bits needed for the algorithm.
Our first use for it is for BP, which is very simple with $d=1$ and $T$ the identity matrix.
BP requires $n$ random bits, which is too many for a practical streaming algorithm.
JLBP, Algorithm~\ref{algo: JLBP}, exists to fix this problem.
Still, if one is willing to disregard the storage needed for the random bits, BP can be substituted everywhere for JLBP without affecting the correctness of our algorithms because our proofs only require that the processes are tame, and BP produces a tame process, as we will now show.
We have a similar lemma for JLBP.

\begin{lemma}[BP Correctness]\label{lem: BP is tame}
  Let $f^{(t)}$, for $t\in[m]$, be the frequency vectors of an insertion-only stream.
  The sequence $Zf^{(t)}$ returned by the algorithm BP is a tame Bernoulli process.
\end{lemma}
\begin{proof}
 By the Chaining Inequality, Theorem~\ref{thm: chaining inequality} below, there exists a constant $C'$ such that $\E\sup_t|X_t|\leq C'(\sum_{j}f_j^2)^{1/2}$.
  Let $F$ be the event that the condition \eqref{eq: tame definition} holds.
  Then for $C\geq C'/\delta$ we have, by Markov's Inequality,
  \iftoggle{shortversion}{
  $\Pr (F) = \Pr \left(\sup_t|X_t|\leq C \sqrt{\sum_j f_j^2}\right)\geq (1-\delta)$.
  }{
  \[\Pr (F) = \Pr \left(\sup_t|X_t|\leq C \sqrt{\sum_j f_j^2}\right)\geq (1-\delta).\]
  }
\end{proof}

In order to reduce the number of random bits needed for the algorithms we first apply JL transformation $T$ to the frequency vector.
The intuition for this comes from the covariance structure of the Bernoulli process, which is what governs the behavior of the process and is fundamental for the Chaining Inequality.
The variance of an increment of the Bernoulli process between times $s$ and $t>s$ is  $\|f^{(s:t)}\|_2^2$.
The JL-property of the matrix $T$ guarantees that this value is well approximated by $\|Tf^{(s:t)}\|_2^2$, which is the increment variance of the reduced-dimension process.
Slepian's Lemma \iftoggle{shortversion}{}{(Lemma~\ref{lem: slepian's})} is a fundamental tool in the theory of Gaussian processes that allows us to draw a comparison between the suprema of the processes by comparing the increment variances instead.
Thus, for $Z_1,\ldots,Z_n\iidsim\Rademacher$, the expected supremum of the process $X_t=\sum_{i=1}^n Z_if_i^{(t)}$ is closely approximated by that of $X_t'= \sum_{i=1}^kZ_i(Tf^{(t)})_i$, and the latter uses only $k=O(\log n)$ random bits.
The following lemma formalizes this discussion, its proof is given in \iftoggle{shortversion}{the full version.}{Section~\ref{sec:randomness}.}
\newcommand{\JLBPLemma}{  
Suppose the matrix $T$ used by JLBP is an $(1\pm\gamma)$-embedding of $(f^{(t)})_{t\in[m]}$.
For any $d\geq 1$, the sequence $\frac{1}{\sqrt{d}}ZTf^{(t)}$ returned by JLBP is a tame $d$-dimensional Bernoulli process.
Furthermore, there exists $d'=O(\log \delta^{-1})$ such that for any $d\geq d'$ and $H\in[n]$ it holds that $\Pr (\frac{1}{2}\leq \|d^{-\frac{1}{2}}ZTe_H\|\leq \frac{3}{2})\geq 1-\delta$.
}
\begin{lemma}[JLBP Correctness]\label{lem: JLBP is tame}
\JLBPLemma
\end{lemma}
\begin{algorithm}
  \begin{algorithmic}
    \Procedure{JLBP}{Stream $S$}
    \State Let $T$ be a JL Transformation \Comment{The same $T$ will suffice for every instance}
    \State Sample $Z\in\{-1,1\}^{d\times k}$ with coordinates $Z_{i,j}\iidsim\Rademacher$
    \State \Return {$\frac{1}{\sqrt{d}} ZTf^{(t)}$ at each time $t$}
    \EndProcedure
  \end{algorithmic}
  \caption{A Bernoulli process with fewer random bits.}\label{algo: JLBP}
\end{algorithm}

Now that we have established the tameness of our Bernoulli processes, let us explain how we can exploit it.
We typically exploit tameness in two ways, one works by splitting the stream according to the items and the second splits the stream temporally.
Given a stream and a tame Bernoulli process on that stream, every substream defines another Bernoulli process, and the substream processes are tame as well.
One way to use this is for heavy hitters.
If there is a heavy-hitter $H$, then the substream consisting of all updates except those to the heavy-hitter produces a tame process whose maximum is bounded by $C(F_2-f_H^2)^{1/2}$, so the value of the process in BP is $Z_Hf_H\pm C(F_2-f_H^2)^{1/2}$.
When $H$ is sufficiently heavy, this means that the absolute value of the output of BP tracks the value of $f_H$, for example if $H$ is a $4C^2$-heavy hitter then the absolute value of BP's output is always a $(1\pm\frac{1}{2})$-approximation to $f_H$.
Another way we exploit tameness is for approximating $F_2$ at all points.
We select a sequnece of times $t_1<t_2<\cdots<t_j\in[m]$ and consider the prefixes of the stream that end at times $t_1,t_2,\ldots,$ etc.
For each $t_i$, the prefix stream ending at time $t_i$ is tame with the upper bound depending on the stream's $F_2$ value at time $t_i$.
If the times $t_i$ are chosen in close enough succession this observation allows us to transform the uniform additive approximation guarantee into a uniform multiplicative approximation.

\subsection{Description of CountSieve}

CountSieve primarily works in two stages that operate concurrently.
Each stage uses independent pairs of Bernoulli processes to determine bits of the identity of the heavy hitter. 
The first stage is the Amplifier, which maintains $L=O(\log\log n)$ independent pairs of Bernoulli processes.
The second stage is the Timer and Sieve.
It consists of a series of rounds where one pair of AMS sketches is maintained during each round.

CountSieve and its subroutines are described formally in Algorithm~\ref{algo: main}.
The random variables they use are
\iftoggle{shortversion}{Amplifier hashes $A_{\ell,1},\ldots,A_{\ell,n}\pwsim\Bernoulli$, for $\ell\in[L]$,
one independent copy of the sequence $Z_{1},\ldots,Z_{k}\iidsim \Rademacher$ for each instance of JLBP, Sieve hashes $B_{r,1},\ldots,B_{r,n}\pwsim\Bernoulli$, and Sieve Rademachers for AMS
$R_{r,1},\ldots,R_{r,n}\fwsim\Rademacher$.
The algorithm also needs a random seed of length $O(\log n)$ bits for the Kane, Meka, Nelson JL generator~\cite{kane2011almost}.
}{
listed in Table~\ref{tab: randomness}.
}
Even though we reduce the number of random bits needed for each Bernoulli process to a managable $O(\log n)$ bits, the storage space for the random values is still an issue because the algorithm maintains $O(\log{n})$ independent hash functions until the end of the stream.
We explain how to overcome this barrier in \iftoggle{shortversion}{the full version}{Section~\ref{sec:randomness}} as well as show that the JL generator of~\cite{kane2011almost} suffices.
\iftoggle{shortversion}{}{
\begin{table}
  \begin{center}
  \begin{tabular}{|c|c|}\hline
$A_{\ell,1},\ldots,A_{\ell,n}\fwsim\Bernoulli$ & 
$Z_{1},\ldots,Z_{k}\iidsim \Rademacher$\\
$B_{r,1},\ldots,B_{r,n}\fwsim\Bernoulli$ &
$R_{r,1},\ldots,R_{r,n}\fwsim\Rademacher$\\\hline
  \end{tabular}
  \end{center}
  \caption{Random vectors for CountSieve.
  Each vector is independent of the others, and $Z=(Z_i)_{i\in[k]}$ is sampled independently for every instance of JLBP.}\label{tab: randomness}
\end{table}
}

\iftoggle{shortversion}{}{We can now state an algorithm that maintains a pair of Bernoulli processes and prove that the bits that it outputs favor the process in the pair with the heavy hitter.}

\begin{algorithm}
  \begin{algorithmic}
    \Procedure{Pair}{Stream $S$, $A_1,\ldots A_n\in\{0,1\}$}
    \State For $b\in\{0,1\}$ let $S_b$ be the restriction of $S$ to $\{j\in[n]\mid A_j=b|\}$
    \State $X^{(t)}_0 = ${\sc JLBP}$(S^{(t)}_0)$ at each time $t$
    \State $X^{(t)}_1 = ${\sc JLBP}$(S^{(t)}_1)$ at each time $t$
    \State $b_t = \argmax_{b\in\{0,1\}} \|X_b^{(t)}\|_2$
    \State \Return {$b_1,b_2,\ldots$}
    \EndProcedure
  \end{algorithmic}
  \caption{Split the vector $f$ into two parts depending on $A$ and run a Bernoulli process on each part.
  Return the identity of the larger estimate at each time.}\label{algo: pair}
\end{algorithm}
\begin{lemma}[Pair Correctness]\label{lem: pair correctness}
  Let $t_0\in[m]$ be an index such that $(f_H^{(t_0)})^2> 4C^2\sum_{j\neq H}f_j^2$.
  Let $A_1,\ldots,A_n\pwsim\Bernoulli$ and $b_1,b_2,\ldots,b_m$ be the sequence returned by Pair$(f,A_1,\ldots,A_n)$.
  Then
\iftoggle{shortversion}{
$\Pr (b_t=A_H\text{ for all }t\geq t_0)\geq 1-3\delta$ and, for every $j\in[n]\setminus\{H\}$ and $t\geq t_0$, $\Pr (b_t=A_j) \leq \frac{1}{2}+3\delta$.
  }{
    \[\Pr (b_t=A_H\text{ for all }t\geq t_0)\geq 1-3\delta\]
  and, for every $j\in[n]\setminus\{H\}$ and $t\geq t_0$,
  \[\Pr (b_t=A_j) \leq \frac{1}{2}+3\delta.\] 
}
Furthermore, if each JLBP is replaced by an AMS sketch with size $O(\log{n}\log\delta^{-1})$ then, for all $t\geq t_0$ and $j\neq H$, $P(b_t=A_H)\geq 1-2\delta$ and $P(b_t=A_j)\leq \frac{1}{2}+3\delta$.
\end{lemma}
\iftoggle{shortversion}{}{
\begin{proof}
Let $X_0^{(t)}=d^{-\frac{1}{2}}ZTf^{(t)}$ and $X_1^{(t)} = d^{-\frac{1}{2}}WTf^{(t)}$ be the two independent Bernoulli processes output by JLBP.
  Without loss of generality, suppose that $A_H=1$, let $v = d^{-\frac{1}{2}}WTe_H$, and let $Y^{(t)} = X_1^{(t)}-f_H^{(t)}v$.
  By Lemma~\ref{lem: JLBP is tame}, with probability at least $1-2\delta$ all three of the following hold
\begin{enumerate}
\item $\|X_0^{(t)}\|^2_2\leq C^2\sum_{j:A_j=0}f_j^2$, for all $t$,
\item $\|Y^{(t)}\|^2_2\leq C^2\sum_{\substack{j\neq H\\A_j=1}} f_j^2$, for all $t$, and
\item $\|v\|_2\geq 1/2$.
\end{enumerate}
If the three events above hold then, for all $t\geq t_0$,
\[\|X_1^{(t)}\|_2-\|X_0^{(t)}\|_2\geq \|v f_H^{(t)}\|_2 -\|Y^{(t)}\|_2 - \|X_0^{(t)}\|_2 \geq \frac{1}{2}f_H^{(t)}-C\sqrt{\sum_{j\neq H} f_j^2}>0,\]
which establishes the first claim.
The second claim follows from the first using
\[\Pr (b_t=A_j) = \Pr (b_t=A_j=A_H) + \Pr (b_t=A_j\neq A_H) \leq \Pr (A_j=A_H) + \Pr (b_t\neq A_H) = \frac{1}{2}+3\delta.\]
The third and fourth inequalities follow from the correctness of the AMS sketch~\cite{ams99}.
\end{proof}
}

\begin{algorithm}
  \begin{algorithmic}
    \Procedure{CountSieve}{Stream $S=(p_1,p_2,\ldots,p_m)$}
    \State Maintain $a_t=(a_{1,t},a_{2,t},\ldots,a_{L,t})\gets${\sc Amplifier}$(S)$
    \State Let $t_1<t_2<\cdots = \{t\in[n] \mid A_{\ell,p_t}=a_{\ell,t}\text{ for at least }0.9L\text{ values of }\ell\}$
    \State Let $S_0 = (p_{t_1},p_{t_2},\ldots,)$
    \State $q_0,q_1, \dots, q_R \gets ${\sc Timer}($S_0$)
    \State $b_1,b_2,\ldots,b_R\gets ${\sc Sieve}$(S_0,  q_0, \dots, q_R)$
    \State \Return {{\sc Selector}(${b_1,b_2,\ldots,b_R}$) based on $S_{0}$}
    \EndProcedure
  \end{algorithmic}
\hrule
  \begin{algorithmic}
    \Procedure{Amplifier}{Stream $S$}\Comment Find a substream where $H$ is $\polylog(n)$-heavy
    \For {$\ell=1,2,\ldots,L$}
    \State $a_{\ell,1},a_{\ell,2},\ldots,a_{\ell,m}\gets${\sc Pair}$(S,A_{\ell,1},\ldots,A_{\ell,n})$
    \EndFor
    \State \Return $a_{1,t}, \dots, a_{L,t}$ at each time $t$ 
    \EndProcedure
  \end{algorithmic}
  \iftoggle{shortversion}{
  \caption{Algorithm for a single $F_2$ heavy hitters.}\label{algo: main}
  \end{algorithm}
  \begin{algorithm}
  }{
  \hrule
  }
  \begin{algorithmic}
    \Procedure{Timer}{Stream $S$}\Comment Break the substream into rounds so $H$ is heavy in each
    \State $q'_0=0$
    \State $Y_t \gets ${\sc JLBP}$(S)$, for $t=1,2,\ldots,$ over $S$
    \State For each $r\geq 1$, find $q'_{r} = \min\{t\mid \|Y_t\|_2 > (1+\frac{1}{\tau})^r\}$
    \State Let $q_0,q_1,\ldots,q_R$ be the last $R+1$ of $q'_0,q'_1,\ldots$
    \State \Return $q_0,q_1,\ldots,q_R$
    \EndProcedure
  \end{algorithmic}
  \hrule
  \begin{algorithmic}
    \Procedure{Sieve}{Stream $S$, $q_0,\dots,q_R$} \Comment Identify one bit of information from each round
    \For {$r=0,1,\ldots,R-1$}
    \State $b_{q_{r}+1},\ldots,b_{q_{r+1}}\gets${\sc Pair}$(S^{(q_{r}:q_{r+1})},B_{r,1},\ldots,B_{r,n})$ \Comment{Replace JLBP here with AMS}
    \EndFor
    \State \Return $b_{q_1},b_{q_2},\ldots,b_{q_R}$
    \EndProcedure
  \end{algorithmic}
  \hrule
  \begin{algorithmic}
    \Procedure{Selector}{$b_1,\ldots,b_R$} \Comment Determine $H$ from the round winners
    \State \Return Any $j^*\in\argmax_j\#\{r \in [R]: B_{r,j}=b_{r}\}$.
    \EndProcedure
  \end{algorithmic}
\iftoggle{shortversion}{}{\caption{Algorithm for a single $F_2$ heavy hitters.}\label{algo: main}}
\end{algorithm}

\subsection{Amplifier correctness}
The $L=O(\log\log n)$ instances of Pair maintained by Amplifier in the first stage of CountSieve serve to identify a substream containing roughly $n2^{-L}=n/\polylog{n}$ elements in which $H$ appears as a $\polylog (n)$-heavy hitter.
Correctness of Amplifier means that, after some ``burn-in'' period which we allow to include the first $f_H/2$ updates to $H$, all of the subsequent updates to $H$ appear in the amplified substream while the majority of other items do not. 
This is Lemma~\ref{lem: amplifier correctness}.

\begin{lemma}[Amplifier Correctness]
\label{lem: amplifier correctness}
  Let $t_0\in[m]$ be such that $(f_H^{(t_0)})^2\geq 4C^2\sum_{j\neq H}f_j^2$, and  let $a_t=(a_{1,t},\ldots,a_{L,t})$ denote the length $L$ bit-vector output by the Amplifier at step $t$.
  Let $M_{j,t} = \#\{\ell\in[L]\mid a_{\ell,t}=A_{\ell,j}\}$ and  $W=\{j\in[n]\setminus\{H\}\mid\exists t\geq t_0, M_{j,t}\geq 0.9L\}$.
  Then, with probability at least $(1-2\delta)$, both of the following hold:
  \iftoggle{shortversion}{  1.\ 
  for all $t\geq t_0$ simultaneously, $M_{H,t}\geq 0.9L$ and 2.\ 
  $\sum_{j\in W}f_j^2 \leq \exp(-\frac{L}{25})\sum_{j\neq H}f_j^2$.
  }{
  \begin{enumerate}
  \item\label{it: amplifier H wins} for all $t\geq t_0$ simultaneously, $M_{H,t}\geq 0.9L$ and
  \item\label{it: amplifier heaviness} $\sum_{j\in W}f_j^2 \leq \exp(-\frac{L}{25})\sum_{j\neq H}f_j^2$.
\end{enumerate}
}
\end{lemma}
\iftoggle{shortversion}{}{
\begin{proof}
Let $N=\#\{\ell\mid \text{for all }t\geq t_0,a_{\ell,t}=A_{\ell,H}\}$.
If $N\geq 0.9L$ then \ref{it: amplifier H wins} holds.
Lemma~\ref{lem: pair correctness} implies $\E N \geq (1-3\delta)L\geq 0.97L$, so Chernoff's Bound easily implies $P(N<0.9L) = O(2^{-L})\leq \delta$, where $\delta$ is a constant.

Now, let $j\neq H$ be a member of $W$ and suppose that $M_{H,t}\geq 0.9L$.
Let $t\geq t_0$ be such that $M_{j,t}\geq 0.9L$. Then it must be that
\[M'_j:=\#\{\ell\in [L]\mid A_{\ell,j}=A_{\ell,H}\} \geq 0.8L.\]
However, $\E M'_j=\frac{1}{2}L$ by pairwise independence.
Let $E_j$ be the event $\{j\in W\text{ and }M_{H,t}\geq0.9L\}$.
Since the $L$ instances of Pair are independent, an application of Chernoff's Inequality proves that $\Pr (E_j)\leq \Pr (M'_j \geq 0.8L)\leq \exp\{\frac{-0.6^2 L}{6}\}\leq e^{-L/20}$.

We have
\[\E (\sum_{j\in W} f_j^2) = \E (\sum_{j\neq H} 1_{E_j}f_j^2) \leq e^{-L/20}\sum_{j\neq H}f_j^2.\]
Therefore Markov's Inequality yields
\[\Pr \left(\sum_{j\in W}f_j^2 \geq e^{-L/25}\sum_{j\neq H}f_j^2\right)\leq e^{-L/100} \leq \delta.\]
The lemma follows by a union bound.
\end{proof}
}

\subsection{Timer and Sieve correctness}

The timing of the rounds in the second stage of CountSieve is determined by Timer.
Timer outputs a set of times $q_0,q_1,\ldots,q_R$ that break the stream into intervals so that each interval has roughly a $1/\log{n}$ fraction of the occurrences of $H$ and not too many other items.
Precisely, we want that $H$ is everywhere heavy for $q$, as stated in the following definition.
When this holds, in every round the Pair is likely to identify one bit of $H$, and Sieve and Selector will be likely to correctly identify $H$ from these bits.

\begin{definition}
Given an item $H\in[n]$ and a sequence of times $q_0<q_1< \cdots< q_R$ in a stream with frequency vectors $(f^{(t)})_{t\in[m]}$ we say that $H$ is \emph{everywhere heavy for $q$} if, for all $1\leq r\leq R$,
\iftoggle{shortversion}{
$(f_H^{(q_{r-1}:q_{r})})^2 \geq C^2 \sum_{j\neq H} (f_j^{(q_{r-1}:q_{r})})^2$.}{
\[(f_H^{(q_{r-1}:q_{r})})^2 \geq C^2 \sum_{j\neq H} (f_j^{(q_{r-1}:q_{r})})^2.\]}
\end{definition}

Correctness for Timer means that enough rounds are completed and $H$ is sufficiently heavy within each round, i.e., $H$ is everywhere heavy for $q$.
\begin{lemma}[Timer Correctness]
  \label{lem: timer is good}
Let $S$ be a stream with an item $H\in[n]$ such that the following hold:
\iftoggle{shortversion}{
1.\ $f_H\geq \tau^4$,
2.\ $f_H^2\geq 400C^2\sum_{j\neq H} f_j^2$, and
3.\ $(f_H^{(t^*:m)})^2=\frac{1}{4}f_H^2\geq 25C^2\tau^2\sum_{j\neq H} (f_j^{(t^*:m)})^2$,
}{
\begin{enumerate}
\item $f_H\geq \tau^4$,
\item $f_H^2\geq 400C^2\sum_{j\neq H} f_j^2$, and
\item\label{it: timer assumption 3} $(f_H^{(t^*:m)})^2=\frac{1}{4}f_H^2\geq 25C^2\tau^2\sum_{j\neq H} (f_j^{(t^*:m)})^2$,
\end{enumerate}
}
where $t^* = \min\{t\in[m]\mid f_H^{(t)}\geq 0.5 f_H\}$ and $C$ is the constant from Definition~\ref{def: tame}.
If $q_0,q_1,\ldots,q_R$ is the sequence output by Timer$(S)$ then, with probability at least $1-4\delta$, $H$ is everywhere heavy for $q$.
\end{lemma}
\iftoggle{shortversion}{}{
\begin{proof}
  We begin by proving that at least $R$ rounds occur after $t^*$, which shows that $q_0,\ldots,q_R$ is well defined, and then we show that $H$ is everywhere heavy.
  Let $Y_t$ be the sequence output by JLBP and let $X_t = Y_t - d^{-\frac{1}{2}}ZTe_Hf_H^{(t)}$.   
$Y_t$ and $X_t$ are tame by Lemma~\ref{lem: JLBP is tame} and $\Pr (0.5\leq \alpha\leq 1.5)\geq 1-\delta$ where $\alpha= \|d^{-\frac{1}{2}}ZTe_H\|_2$.
  Hereafter, we suppose that $\alpha\geq 1/2$ and the tameness property holds for $Y_t$ and $X_t$.
  With probability at least $1-\delta$, simultaneously for all $t\in[m]$, we have
  \begin{equation}\label{eq: timer tameness}
\|X_t\|_2^2\leq C^2\sum_{j\neq H}f_j^2\leq \frac{1}{400}f_H^2.
\end{equation}

  Therefore, $\|Y_{t^*}\|_2\leq  \|X_{t^*}\|_2  + \alpha f_H^{(t^*)} \leq (\frac{\alpha}{2}  + \frac{1}{20})f_H$ and $\|Y_{m}\|_2\geq \alpha f_H^{(m)} - \|X_{m}\|_2 \geq (\alpha  - \frac{1}{20})f_H$. 
  This implies that the number of rounds completed after $t^*$, which is 
  $$\log_{1+1/\tau}\frac{\|Y_m\|_2}{\|Y_{t^*}\|_2} \ge \log_{1+1/\tau}\frac{\alpha - 1/20}{\alpha/2 + 1/20} \ge \log_{1+1/\tau}(3/2),$$ is at least $R+1$ by our choice of $\tau=100(R+1)$. Similarly $\|Y_{t^*}\|_2\ \geq  \alpha f_H^{(t^*)} - \|X_{t^*}\|_2 \geq (\frac{\alpha}{2}  - \frac{1}{20})f_H$. Therefore we also get $q_i>q_{i-1}$ because $(1+\tau^{-1})\|Y_{t^*}\|_2\geq 1$ by our assumption that $f_H\geq \tau^4$. Hence $q_0,\ldots,q_R$ are distinct times.

Now we show that $H$ is everywhere heavy for $q$.
Let $W_t = X_t-X_{t^*}$, for $t\geq t^*$.
By design, $W_t - W_s = X_t-X_s$, for $s,t\geq t^*$.
By Lemma~\ref{lem: JLBP is tame}, $W_t$ is also a tame process on the suffix of the original stream that has its first item at time $t^*+1$.
Specifically with probability at least $1-\delta$, for all $t\geq t^*$,
\[
\|W_t\|_2^2 \leq C^2 \sum_{j\neq H} (f_j^{(t^*:m)})^2\leq \frac{1}{400\tau^2}f_H^2.
\]
This inequality, with two applications of the triangle inequality, implies
\begin{equation}\label{eq: f Y W triangle}
\alpha f_H^{(q_{i-1}:q_i)} \geq \|Y_{q_i}-Y_{q_{i-1}}\|_2 - \|W_{q_i}-W_{q_{i-1}}\|_2\geq \|Y_{q_i}-Y_{q_{i-1}}\|_2 - \frac{2}{20\tau}f_H.
\end{equation}
To complete the proof we must bound $\|Y_{q_i}-Y_{q_{i-1}}\|_2$ from below and then apply the heaviness, i.e., assumption~\ref{it: timer assumption 3}.

  Equation~\eqref{eq: timer tameness} and the triangle inequality imply that, for every $t\geq t^*$, it holds that $\|Y_{t}\|_2\geq \alpha f_H^{(t)} - \|X_{t}\|_2 \geq (\frac{\alpha}{2}  - \frac{1}{20})f_H$.
Recalling the definition of $q_0', q_1', \cdots$ from Timer Procedure, since $t^*\leq q_0<q_1<\cdots<q_R$ and the rounds expand at a rate $(1+1/\tau)$,
\begin{equation}\label{eq: plenty of Hs}
\|Y_{q_{i+1}}-Y_{q_i}\|_2\geq \frac{1}{\tau}\left( \frac{\alpha}{2} - \frac{1}{20}\right) f_H.
\end{equation}
Using what we have already shown in~\eqref{eq: f Y W triangle} we have
\[\alpha f_H^{(q_i:q_{i+1})} \geq \frac{1}{\tau}\left( \frac{\alpha}{2} - \frac{1}{20} - \frac{2}{20}\right) f_H\]
so dividing and using $\alpha\geq 1/2$ and $C$ sufficiently large we get \[(f_H^{(q_i:q_{i+1})})^2 \geq \frac{1}{25\tau^2}f_H^2\geq C^2 \sum_{j\neq H}(f^{(t^*:m)}_j)^2\geq C^2 \sum_{j\neq H}(f^{(q_{i}:q_{i+1})}_j)^2.\]
Since this holds for all $i$, $H$ is everywhere heavy for $q$.
We have used the tameness of the three processes ($X$, $Y$, and $W$) and the bounds on $\alpha$.  
Each of these fails with probability at most $\delta$, so the total probability that Timer fails to achieve the condition that $H$ is everywhere heavy for $q$ is at most $4\delta$.
\end{proof}
}

During each round, the algorithm Sieve uses a hash function $A$ to split the stream into two parts and then determines which part contains $H$ via Pair. 
For these instances of Pair, we replace the two instances of JLBP with two instances of AMS. This replacement helps us to hold down the storage when we later use Nisan's PRG, because computing the JL transformation~$T$ from~\cite{kane2011almost} requires $O(\log n\log\log n)$ bits.
Applying Nisan's PRG to an algorithm that computes entries in $T$ would leave us with a bound of $O(\log{n}(\log\log n)^2)$.  More details can be found in \iftoggle{shortversion}{the full version.}{Section~\ref{sec:randomness}.}

A total of $O(\log n)$ rounds is enough to identify the heavy hitter and the only information that we need to save from each round is the hash function $A$ and the last bit output by Pair.
Selector does the work of finally identifying $H$ from the sequence of bits output by Sieve and the sequence of hash functions used during the rounds.
We prove the correctness of Sieve and Selector together in the following lemma.
\begin{lemma}[Sieve/Selector]\label{lem: Selector outputs HH}
  Let $q_0,q_1,\ldots,q_R=\textsc{Timer}(S)$ and let $b_1,\ldots,b_R=\textsc{Sieve}(S,q_0,\ldots,q_R)$.
  If $H$ is everywhere heavy for $q$ on the stream $S$ then, with probability at least $1-\delta$, Selector$(b_1,\ldots,b_R)$ returns $H$.
\end{lemma}
\iftoggle{shortversion}{}{
\begin{proof}
  Lemma~\ref{lem: pair correctness} in the AMS case implies that the outcome of round $r$ satisfies $\Pr (b_r = B_{r,H})\geq 1-3\delta$ and $\Pr (b_r=B_{r,j})\leq \frac{1}{2}+3\delta$.
  The random bits used in each iteration of the for loop within Sieve are independent of the other iterations.
  Upon choosing the number of rounds $R=O(\log n)$ to be sufficiently large, Chernoff's Inequality implies that, with high probability, $H$ is the unique item in $\argmax_j\#\{r\in[R]\mid B_{r,j}=b_r\}$.
Therefore, Selector returns $H$.
\end{proof}
}

\begin{algorithm}
  \begin{algorithmic}
    \Procedure{$\ell_2$HeavyHitters}{Stream $S=(p_1,p_2,\ldots,p_m)$}
    \State $Q\gets O(\log\epsilon^{-1})$, $B\gets O(\epsilon^{-2})$
    \State Select indep.\ 2-universal hash functions 
    \State \hspace{20pt}$h_1,\ldots,h_Q,h'_1,\ldots,h_Q':[n]\to [B]\text{ and }\sigma_1,\ldots,\sigma_Q:[n]\to\{-1,1\}.$
    \State $\hat{F}_2\gets (1\pm\frac{\epsilon}{10})F_2$ using AMS~\cite{ams99}
    \State $\hat{\calH}\gets\emptyset$
    \For{$(q,b)\in Q\times B$}
    \State Let $S_{q,b}$ be the stream of items $i$ with $h_{q}(i)=b$
    \State $c_{q,b}\gets \sum_{j:h_q'(j)=b}\sigma_q(j)f_j$ \Comment{The CountSketch~\cite{ccf04}}
    \State $H\gets \textsc{CountSieve}(S_{q,b})$
    \EndFor
    \State Remove from $\hat{\calH}$ any item such that $i$ such that $\text{median}_q\{|c_{q,h_q(i)}|\}\leq \frac{3\epsilon}{4}\hat{F}_2$.
    \State \Return {$\hat{\calH}$}
    \EndProcedure
  \end{algorithmic}
  \caption{$\ell_2$ heavy hitters algorithm.}\label{algo: F2heavy}
\end{algorithm}
\subsection{CountSieve correctness}

We now have all of the pieces in place to prove that CountSieve correctly identifies a sufficiently heavy heavy hitter $H$.
As for the storage bound and Theorem~\ref{thm:main}, the entire algorithm fits within $O(\log n\log\log n)$ bits except the $R=O(\log n)$ hash functions required by Sieve.
We defer their replacement to \iftoggle{shortversion}{the full version.}{Theorem~\ref{thm: sieve space} in Section~\ref{sec:randomness}.}
\begin{theorem}[CountSieve Correctness]\label{thm: countsieve works}
  If $H$ is a $400C^2$-heavy hitter then, with probability at least $0.95$ CountSieve returns $H$.
The algorithm uses $O(\log{n}\log\log{n})$ bits of storage and can be implemented with $O(\log n\log\log n)$ stored random bits.
\end{theorem}
\iftoggle{shortversion}{}{
\begin{proof}
  We use Theorem~\ref{thm: meka JL} to generate the JL transformation $T$.
  Each of our lemmas requires that $T$ embeds a (possible different) polynomially sized set of vectors, so, for $\delta=\Omega(1)$, Theorem~\ref{thm: meka JL} implies that, with probability at least $1-\delta$, $T$ embeds all of the necessary vectors with seed length $O(\log{n})$, and the entries in $T$ can be computed in space $O(\log{n}\log\log{n})$ bits of space.
  Because of the heaviness assumption, the conclusion of Lemma~\ref{lem: amplifier correctness} fails to hold for $t_0=t^*$ (defined in Lemma~\ref{lem: timer is good}) with probability at most $2\delta$.
  When that failure does not occur, the second and third hypotheses in Lemma~\ref{lem: timer is good} hold.
  The first hypothesis is that $f_H\geq\tau^4$, suppose it holds.
  Then the probability that $H$ fails to be everywhere heavy for the sequence $q$ that is output by Timer is at most $4\delta$.
  In this case, according to Lemma~\ref{lem: Selector outputs HH}, Sieve and Selector correctly identify $H$ except with probability at most $\delta$.
  Therefore, the algorithm is correct with probability at least $1-8\delta\geq 0.95$, by choosing $\delta \leq 1/200$. 
  If $f_H<\tau^4$, then because $H$ is a heavy hitter, we get $\sum_{j\neq H}f_j^2 \leq \tau^8 = O(\log^8 n)$.
  Then we choose the constant factor in $L$ large enough so that, the second conclusion of Lemma~\ref{lem: amplifier correctness} implies $\sum_{j\in W}f_j^2\leq e^{-L/25}<1$.
  This means that $H$ is the only item that passes the amplifier for all $t\geq t^*$, and, no matter what is the sequence output by Timer, $H$ is everywhere heavy because it is the only item in the substream.
 Thus, in this case the algorithm also outputs $H$.

  Now we analyze the storage and randomness.
Computing entries in the Kane-Meka-Nelson JL matrix requires $O(\log n\log\log n)$ bits of storage, by Theorem~\ref{thm: meka JL}, and there is only one of these matrices.
Amplifier stores $L=O(\log\log n)$ counters. Sieve, Timer, and Selector each require $O(\log n)$ bits at a time (since we discard any value as soon as it is no longer needed).
Thus the total working memory of the algorithm is $O(\log n\log\log n)$ bits.
The random seed for the JL matrix has $O(\log n)$ bits.
Each of the $O(\log\log n)$ Bernoulli processes requires $O(\log n)$ random bits.
By Theorem~\ref{thm: sieve space} below, the remaining random bits can be generated with Nisan's generator using a seed of $O(\log n\log\log n)$ bits.
Using Nisan's generator does not increase the storage of the algorithm.
Accounting for all of these, the total number of random bits used by CountSieve, which also must be stored, is $O(\log n\log\log n)$.
Therefore, the total storage used by the algorithm is $O(\log n\log\log n)$ bits.
\end{proof}
}

\newtheorem*{hhtheorem}{Theorem~\ref{thm:main}}
\begin{hhtheorem}[$\ell_2$-Heavy Hitters]
\heavyhitterstheorem
\end{hhtheorem}
\iftoggle{shortversion}{}{
\begin{proof}
The algorithm is Algorithm~\ref{algo: F2heavy}.
It has the form of a CountSketch~\cite{ccf04} with $Q=O(\log 1/\epsilon)$ ``rows'' and $B=8(10 C)^2/\epsilon^2$ ``buckets'' per row, wherein we run one instance of CountSieve in each bucket to identify potential heavy hitters and also the usual CountSketch counter in each bucket.
Finally, the algorithm discriminates against non-heavy hitters by testing their frequency estimates from the CountSketch.
We will assume that the AMS estimate $\hat{F}_2$ is correct with probability at least $8/9$.

Let $\calH_k = \{i\mid f_i\geq \frac{\epsilon}{k}\sqrt{F_2}\}$ and let $\hat{\calH}$ be set of distinct elements returned by Algorithm~\ref{algo: F2heavy}.
To prove the theorem, it is sufficient to prove that, with probability at least $2/3$, $\calH_1\subseteq \hat{\calH}\subseteq \calH_2$.

Let $H\in \calH_1$ and consider the stream $S_{q,h_q(H)}$ at position $(q,h_q(H))$.
We have
\[\E(\sum_{\substack{j\neq H \\ h_q(j)=h_q(H)}}f_j^2)\leq \frac{\epsilon^2}{8(10C)^2}F_2.\]
Let $E_{q,H}$ be the event that
\[\sum_{\substack{j\neq H \\ h_q(j)=h_q(H)}}f_j^2\leq \frac{\epsilon^2}{(10C)^2}F_2,\]
so by Markov's Inequality $\Pr(E_{q,H})\geq 7/8$.
When $E_{q,H}$ occurs $H$ is sufficiently heavy in $S_{q,h_q(H)}$ for CountSieve.
By Theorem~\ref{thm: countsieve works}, with probability at least $\frac{7}{8}-\frac{1}{20}\geq 0.8$, CountSieve identifies $H$.
Therefore, with the correct choice of the constant factor for $Q$, a Chernoff bound and a union bound imply that, with probability at least $1-1/9$, every item in $\calH_1$ is returned at least once by a CountSieve.

Let $\hat{\calH}'$ denote the set $\hat{\calH}$ before any elements are removed in the final step.
Since CountSieve identifies at most one item in each bucket, $|\hat{\calH}'|=O(\epsilon^{-2}\log\epsilon^{-1})$. 
By the correctness of CountSketch~\cite{ccf04} and the fact that it is independent of $\hat{H}'$, we get that, with probability at least $1-1/9$, for all $i\in\hat{H}'$ 
\[\left|f_i - \text{median}_q\{|c_{q,h_q(i)}|\}\right|\leq \frac{\epsilon}{10C}\sqrt{F_2}.\]
When this happens and the AMS estimate is correct, the final step of the algorithm correctly removes any items $i\notin\calH_2$ and all items $i\in\calH_1$ remain. 
This completes the proof of correctness.

The storage needed by the CountSketch is $O(BQ\log{n})$, storage needed for the CountSieves is $O(BQ\log{n}\log\log{n})$, and the storage needed for AMS is $O(\epsilon^{-2}\log{n})$.
Therefore the total storage is $O(BQ\log{n}\log\log{n}) = O(\frac{1}{\epsilon^2}\log{\frac{1}{\epsilon}}\log{n}\log\log{n})$ bits.
\end{proof}
}

\begin{corollary}
There exists an insertion-only streaming algorithm that returns an additive $\pm\epsilon\sqrt{F_2}$ approximation to $\ell_\infty$, with probability at least $2/3$.
The algorithm requires $O(\frac{1}{\epsilon^2}\log{\frac{1}{\epsilon}}\log{n}\log\log{n})$ bits of space.
\end{corollary}
\iftoggle{shortversion}{}{
\begin{proof}
Use Algorithm~\ref{algo: F2heavy}.
If no heavy-hitter is returned then the $\ell_\infty$ estimate is 0, otherwise return the largest of the CountSketch medians among the discovered heavy hitters.
The correctness follows from Theorem~\ref{thm:main} and the correctness of CountSketch.
\end{proof}
}

%% file: Sec-Chaining.tex
We call these inequalities Chaining Inequalities after the Generic Chaining, which is the technique that we use to prove it.
The book~\cite{talagrand2014upper} by Talagrand contains an excellent exposition of the subject.
\iftoggle{shortversion}{
The full version contains a more detailed discussion of the method.
}{
Let $(X_t)_{t\in T}$ be a Gaussian process.
The Generic Chaining technique concerns the study of the supremum of $X_t$ in a particular metric space related to the variances and covariances of the process.
The metric space is $(T,d)$ where $d(s,t) = (\E(X_s-X_t)^2)^{\frac{1}{2}}$.
The method takes any finite chain of finite subsets $T_0\subseteq T_1\subseteq\cdots\subseteq T_n\subseteq T$ and uses $(X_t)_{t\in T_i}$ as a sequence of successive approximations to $(X_t)_{t\in T}$ wherein $X_t$, for $t\notin T_i$, is approximated by the value of the process at some minimizer of $d(t,T_i)=\min\{d(t,s)\mid s\in T_i\}$.
To apply the Generic Chaining one must judiciously choose the chain in order to get a good bound, and the best choice necessarily depends on the structure of the process.
We will exploit the following lemma.
}
\begin{lemma}[\cite{talagrand2014upper}]\label{lem: chain of sets}
Let $\{X_t\}_{t\in T}$ be a Gaussian process and let $T_0\subseteq T_1 \dots \subseteq T_n \subseteq T$  be a chain of sets such that
$|T_0|=1$ and $|T_i|\le 2^{2^i}$ for $i\ge 1$.
Then
\iftoggle{shortversion}{
$\E  \sup_{t\in T} X_t \le O(1) \sup_{t\in T} \sum_{i\ge 0}2^{i/2}d(t, T_i)$.
}{
\begin{equation}\label{chain}
\E  \sup_{t\in T} X_t \le O(1) \sup_{t\in T} \sum_{i\ge 0}2^{i/2}d(t, T_i).
\end{equation}
}
\end{lemma}

The Generic Chaining also applies to Bernoulli processes, but, for our purposes, it is enough that we can compare related Gaussian and Bernoulli processes.
\begin{lemma}[\cite{talagrand2014upper}]\label{lem: Gaussion to Bernoulli}
  Let $A\in\R^{m\times n}$ be any matrix and let $G$ and $B$ be $n$-dimensional vectors with independent  coordinates distributed as $N(0,1)$ and $\Rademacher$, respectively.
  Then the Gaussian process $X = AG$ and Bernoulli process $Y=AB$ satisfy
  \iftoggle{shortversion}{
 $\E  \sup_{t\in T} Y_t \le  \sqrt{\frac{\pi}{2}}\E  \sup_{t\in T} X_t.$
  }{
\begin{equation*}
 \E  \sup_{t\in T} Y_t \le  \sqrt{\frac{\pi}{2}}\E  \sup_{t\in T} X_t.
\end{equation*}
}
\end{lemma}
\begin{theorem}[Chaining Inequality]\label{thm: chaining inequality}
  Let $Z_1,\ldots,Z_n\ldots\iidsim\calN(0,1)$ and let $(f^{(t)})_{t\in[m]}$ be the sequence of frequency vectors of an insertion-only stream.
  There exists a universal constant $C'>0$ such that if $X_t = \sum_{j=1}^n Z_jf_{j}^{(t)}$, for $t\in[m]$, then
 \iftoggle{shortversion}{$\E \sup_i|X_i| \leq C'\sqrt{\Var(X_m)} = C'\|f^{(m)}\|_2$.
 }{
\begin{equation}\label{eq: gaussian chaining}
\E \sup_i|X_i| \leq C'\sqrt{\Var(X_m)} = C'\|f^{(m)}\|_2.
\end{equation}
}
  If $\bar Z_1,\ldots,\bar Z_n\ldots\iidsim\Rademacher $ and 
   $Y_t = \sum_{j=1}^n \bar Z_jf_{j}^{(t)}$, for $t\in[m]$, then
 \iftoggle{shortversion}{$\E \sup_i|Y_i| \leq C'\sqrt{\Var(Y_m)}= C'\|f^{(m)}\|_2$.
 }{
 \begin{equation}\label{eq: bernoulli chaining}
\E \sup_i|Y_i| \leq C'\sqrt{\Var(Y_m)}= C'\|f^{(m)}\|_2.
\end{equation}
}
\end{theorem}
\begin{proof}
Let $T=[m]$.
Define $T_0 = \{t_0\}$, where $t_0$ is the index such that $\Var(X_{t_0}) < 0.5\Var(X_{m}) \le \Var(X_{t_0 + 1})$ and $T_i = \{1, t_{i,1}, t_{i,2},\dots\}$ where for each index $t_{i,j}\in T_i$ $\Var(X_{t_{i,j}}) < \frac{j}{2^{2^i}}\Var(X_{m}) \le \Var(X_{t_{i,j + 1}})$.
This is well-defined because $\Var(X_{t})=\|f^{(t)}\|_2^2$ is the second moment of an insertion-only stream, which must be monotonically increasing.
By construction $|T_i| \le 2^{2^i}$ and, for each $t\in T$, there exist $t_{i,j}\in T_j$ such that $d(t, T_i) = \min(d(t, t_{i,j}), d(t, t_{i,j+1})) \le d(t_{i,j}, t_{i,j+1}) = (\E(X_{t_{i,j+1}}-X_{t_{i,j}})^2)^{\frac{1}{2}}$, where the last inequality holds because $E(X_t^2)$ monotonically increasing with $t$.

Notice that every pair of increments has nonnegative covariance because the stream is insertion-only.
Thus, the following is true:
\iftoggle{shortversion}{
\begin{align*}
d(t,t_{i,j+1})^2 &\le \E(X_{t_{i,j+1}}-X_{t_{i,j}})^2 \leq \E(X_{t_{i,j+1}}-X_{t_{i,j}})^2 + 2 \E X_{t_{i,j}}(X_{t_{i,j+1}}-X_{t_{i,j}})\\
& = \E X_{t_{i,j+1}}^2 -\E X_{t_{i,j}}^2 \le\frac{j+1}{2^{2^i}}\E X_{m}^2 - \frac{j-1}{2^{2^i}}\E X_m^2 = \frac{2}{2^{2^i}}\E X_m^2.
\end{align*}
}{
\begin{align*}
d(t,t_{i,j+1})^2 &\le \E(X_{t_{i,j+1}}-X_{t_{i,j}})^2\\
& \leq \E(X_{t_{i,j+1}}-X_{t_{i,j}})^2 + 2 \E X_{t_{i,j}}(X_{t_{i,j+1}}-X_{t_{i,j}})\\
& = \E X_{t_{i,j+1}}^2 -\E X_{t_{i,j}}^2 \\
&\le\frac{j+1}{2^{2^i}}\E X_{m}^2 - \frac{j-1}{2^{2^i}}\E X_m^2 = \frac{2}{2^{2^i}}\E X_m^2.
\end{align*}
}
Then we can conclude that
\iftoggle{shortversion}{
$\sum_{i\ge 0}2^{i/2}d(t, T_i) \le \sum_{i\ge 0}2^{i/2}\frac{2}{2^{2^{i}}}\sqrt{\E X_m^2} = O(1)\sqrt{\Var(X_m)}$.
}{
\[\sum_{i\ge 0}2^{i/2}d(t, T_i) \le \sum_{i\ge 0}2^{i/2}\frac{2}{2^{2^{i}}}\sqrt{\E X_m^2} = O(1)\sqrt{\Var(X_m)}.\]
}
Applying \iftoggle{shortversion}{Lemma~\ref{lem: chain of sets}}{ineqality \eqref{chain}} we obtain
$\E  \sup_{t\in T} X_t \le O(1) \sqrt{\Var(X_m)}$. 

In order to bound the absolute value, observe
\begin{equation}\label{eq: sup vs absolute sup}
\sup_t|X_t|\leq |X_1|+\sup |X_t-X_1|\leq|X_1|+\sup_{s,t}(X_t-X_s)=|X_1| + \sup_tX_t + \sup_s(-X_s).
\end{equation}
Therefore, $\E \sup_t|X_t|\leq \E |X_1| + 2\E \sup X_t\leq O(1)\sqrt{\Var(X_m)}$, because $-X_t$ is also Gaussian process with the same distribution as $X_t$ and $\E |X_1|=O(\sqrt{\Var(X_m)})$ because $f^{(1)} = 1$.
This establishes \iftoggle{shortversion}{the Gaussian inequality and the Bernoulli inequality}{\eqref{eq: gaussian chaining} and \eqref{eq: bernoulli chaining}} follows immediately by an application of Lemma~\ref{lem: Gaussion to Bernoulli}.
\end{proof}

Theorem~\ref{thm: chaining inequality} would obviously \emph{not} be true for a stream with deletions, since we may have $\Var(X_m)=0$.
One may wonder if the theorem would be true for streams with deletions upon replacing $\Var(X_m)$ by $\max_t \Var(X_t)$.
This is not true, and a counter example is the stream $(e_1,-e_1,e_2,\ldots,e_n,-e_n)$ which yields $\max_t\Var(X_t) = 1$, but $\E\sup_t|X_t| =\Theta(\sqrt{\log{n}})$.

Theorem~\ref{thm: chaining inequality} does not apply to the process ouput by JLBP, but the covariance structures of the two processes are very similar because $T$ is an embedding.
\iftoggle{shortversion}{
In the full version, we prove basically the same inequality for the JLBP process by mimicking the stategy in~\cite{meka2012ptas}.
}{
We can achieve basically the same inequality for the JLBP process by applying Slepian's Lemma, mimicking the stategy in~\cite{meka2012ptas}.
\begin{lemma}[Slepian's Lemma~\cite{ledoux1991probability}]\label{lem: slepian's}
Let $X_t$ and $Y_t$, for $t\in T$, be Gaussian processes such that $\E  (X_s-X_t)^2\leq \E (Y_s-Y_t)^2$, for all $s,t\in T$.
Then, $\E \sup_{t\in T} X_t\leq \E \sup_{t\in T} Y_t$.
\end{lemma}
\begin{corollary}[Chaining Inequality for the transformed stream]\label{cor: JL chaining}
  Let $T$ be a $(1\pm\gamma)$-embedding of $(f^{(t)})_{t\in [m]}$ and let $Z_1,\ldots,Z_k\ldots\iidsim\calN(0,1)$. 
    There exists a universal constant $C'>0$ such that if $X_t = \langle Z, Tf^{(t)}\rangle $, for $t\in[m]$, then $\E \sup_i|X_i| \leq C'\|f^{(m)}\|_2$.
 If $\bar Z_1,\ldots,\bar Z_k\iidsim\Rademacher $ and $Y_t = \langle \bar{Z},Tf^{(t)}\rangle $, for $t\in[m]$, then $\E \sup_i|Y_i| \leq C'\|f^{(m)}\|_2$.
\end{corollary}
\begin{proof}
  Let $W_t$ be the Gaussian process from Theorem~\ref{thm: chaining inequality}.
  Since $T$ is a JL transformation 
\[\E (X_t-X_s)^2 = \|Tf^{(s:t)}\|_2^2 \leq (1+\gamma)^2\|f^{(s:t)}\|_2^2 = (1+\gamma)^2 \E(W_t-W_s)^2.\]
The first claim of the corollary follows from Slepian's Lemma, Equation~\eqref{eq: sup vs absolute sup}, and Theorem~\ref{thm: chaining inequality}.
The second inequality follows from the first and Lemma~\ref{lem: Gaussion to Bernoulli}.
\end{proof}
}

%% file: Sec-Derandomization.tex
This section describes how CountSieve can be implemented with only $O(\log n\log\log n)$ random bits.
There are two main barriers to reducing the number of random bits.
We have already partially overcome the first barrier, which is to reduce the number of bits needed by a Bernoulli process from $n$, as in the algorithm BP, to $O(\log n)$ by introducing JLBP.
JLBP runs $d=O(1)$ independent Bernoulli processes in dimension $k=O(\log n)$ for a total of $dk=O(\log n)$ random bits.
This section proves the correctness of that algorithm.

The second barrier is to find a surrogate for the $R=O(\log{n})$ independent vectors of pairwise independent Bernoulli random variables that are used during the rounds of Sieve.
We must store their values so that Selector can retroactively identify a heavy hitter, but, na\"{i}vely, they require $\Omega(\log^2 n)$ random bits.
We will show that one can use Nisan's pseudorandom generator (PRG) with a seed length of $O(\log n \log\log n)$ bits to generate these vectors.
A priori, it is not obvious that this is possible.
The main sticking point is that the streaming algorithm that we want to derandomize must store the random bits it uses, which means that these count against the seed length for Nisan's PRG.
Specifically, Nisan's PRG reduces the number of random bits needed by a space $S$ algorithm using $R$ random bits to $O(S\log R)$.
Because CountSieve must pay to store the $R$ random bits, the storage used is $S\geq R = \Omega(\log^2 n)$, so Nisan's PRG appears even to increase the storage used by the algorithm!
We can overcome this by introducing an auxiliary (non-streaming) algorithm that has the same output as Sieve and Selector, but manages without storing all of the random bits.
This method is similar in spirit to Indyk's derandomization of linear sketches using Nisan's PRG~\cite{i06}.
It is not a black-box reduction to the auxiliary algorithm and it is only possible because we can exploit the structure of Sieve and Selector.

We remark here that we are not aware of any black-box derandomization of the Bernoulli processes that suits our needs.
This is for two reasons.
First, we cannot reorder the stream for the purpose of the proof because the order of the computation is important.
Reordering the stream is needed for Indyk's argument~\cite{i06} for applying Nisan's PRG.
Second, the seed length of available generators is too large, typically in our setting we would require a seed of length at least $\log^{1+\delta}n$, for some $\delta>0$. 

\subsection{The Bernoulli process with $O(\log{n})$ random bits}

The main observation that leads to reducing the number of random bits needed by the algorithm is that the distribution of the corresponding Gaussian process depends only on the second moments of the increments.
These moments are just the square of the Euclidean norm of the change in the frequency vector, so applying a Johnson-Lindenstrauss transformation to the frequency vector nearly preserves the distribution of the process and allows us to get away with $O(\log n)$ random bits.
One trouble with this approach is that the heavy hitter $H$ could be ``lost'', whereby we mean that although $\|Te_H\|\approx 1$ it may be that $\langle Z,Te_H\rangle\approx 0$, for the Rademacher random vector $Z$, whereupon $H$'s contribution to the sum $\langle Z, Tf^{(t)}\rangle$ is lost among the noise.
To avoid this possibility we keep $d=O(1)$ independent Bernoulli processes in parallel.

First, we state the correctness of the Johnson-Lindenstrauss transformation that we use and the storage needed for it.
\begin{theorem}[Kane, Meka, \& Nelson~\cite{kane2011almost}]\label{thm: meka JL}
Let $V=\{v_1,\ldots,v_n\}\subseteq \R^n$.
For any constant $\delta>0$ there exists a $k=O(\gamma^{-2}\log (n/\delta)$ and generator $G:\{0,1\}^{O(\log n)}\times[k]\times[n]\to\R$ such that, with probability at least $1-\delta$, the $k\times n$ matrix $T$ with entries $T_{ij}=G(R,i,j)$ is a $(1\pm\gamma)$-embedding of $V$, where $R\in\{0,1\}^{O(\log{n})}$ is a uniformly random string.
The value of $G(R,i,j)$ can be computed with $O(\log n\log\log n)$ bits of storage.
\end{theorem}

\newtheorem*{lemmaJLBP}{Lemma~\ref{lem: JLBP is tame}}
\begin{lemmaJLBP}[JLBP Correctness]
\JLBPLemma
\end{lemmaJLBP}
\begin{proof}
Let $X_{i,t} =  \sum_{j=1}^kZ_{ij}(Tf^{(t)})_j$ and
\[X_t =  \|\frac{1}{\sqrt{d}}ZTf^{(t)}\|_2^2 = \frac{1}{d}\sum_{i=1}^dX^2_{i,t},\]
for $t=1,\ldots,m$.
Each process $X_{i,t}$ is a Bernoulli process with $\Var(X_{i,t}) = \|Tf^{(t)}\|^2_2\leq (1+\gamma)^2\|f^{(t)}\|^2_2$ and, for $s<t$, $\E (X_{i,t}-X_{i,s})^2 = \|Tf^{(s:t)}\|^2_2\leq (1+\gamma)^2\|f^{(s:t)}\|_2^2$.

Notice that for all $i$ Gaussian processes $(X_{i,t})_{t\in[m]}$ are from same distribution. Let $X_{t}'$ be a Gaussian process that is identical to $X_{i,t}$, except that the Rademacher random variables are replaced by standard Gaussians.
$X_t'$ and $X_{i,t}$ have the same means, variances, and covariances.
Therefore $\E \sup_t|X_{i,t}|\leq\sqrt{\frac{\pi}{2}}\E \sup_t|X'_t|$, by Lemma~\ref{lem: Gaussion to Bernoulli}.

Let $N_1,\ldots,N_n\iidsim N(0,1)$. We will compare $X'_{t}$ against the Gaussian process $X''_t = (1+\gamma)\frac{1}{\sqrt{d}}\langle N,f^{(t)}\rangle$.
By the Chaining Inequality, there exists $C'$ such that $\E \sup |X''_t|\leq C' \sqrt{\Var(X''_m)}=\frac{C'(1 + \gamma)}{\sqrt{d}}\|f^{(m)}\|_2$.
We have $\E (X''_t-X''_s)^2 =\frac{1}{d}(1+\gamma)^2\|f^{(s:t)}\|_2^2$, so by Slepian's Lemma applied to $X'_t$ and $X''_t$ and by \eqref{eq: sup vs absolute sup} we have
\[\E \sup_t|X_{i,t}|\leq \sqrt{\frac{\pi}{2}}\E \sup |X'_t|\leq \sqrt{\frac{\pi}{2}} \sqrt{d}\E\sup_t|X''_t|\leq \sqrt{\frac{\pi}{2}} (1 + \gamma) C'\|f^{(m)}\|_2.\]
Now we apply Markov's Inequality to get $\Pr (\sup_t|X_{i,t}|\geq \frac{C}{\sqrt{d}}\|f^{(m)}\|_2)\leq \frac{\delta}{d}$, by taking $C\geq \sqrt{\frac{\pi}{2}}(1 + \gamma) C'd^{3/2}/\delta$.
From a union bound we find $\Pr (\sup_{i,t}|X_{i,t}|\geq \frac{C}{\sqrt{d}}\|f^{(m)}\|_2)\leq \delta$, and that event implies $\sup_t|X_t|\leq C\|f^{(m)}\|_2$, which is \eqref{eq: tame definition} and proves that the process is tame.

For the second claim, we note that the matrix $\frac{1}{\sqrt{d}}Z$ is itself a type of Johnson-Lindenstrauss transformation (see \cite{achlioptas2003database}), hence $\frac{1}{2}\leq \|d^{-1/2}ZTe_H\|\leq \frac{3}{2}$, with probability at least $1-2^{-d}\geq (1-\delta)$.
The last inequality follows by our choice of $d$.
\end{proof}

\subsection{Sieve and Selector}

In the description of the algorithm, the Sieve and Selector use $O(\log{n})$ many pairwise independent hash functions that are themselves independent.
Nominally, this needs $O(\log^2n)$ bits.
However, as we show in this section, it is sufficient to use Nisan's pseudorandom generator~\cite{n92} to generate the hash functions.
This reduces the random seed length from $O(\log^2n)$ to $O(\log n\log\log n)$.
Recall the definition of a pseudorandom generator.
\begin{definition}
A function $G : \{0,1\}^m \rightarrow \{0, 1\}^n$ is called a \emph{pseudorandom generator (PRG) for space($S$) with parameter $\epsilon$} if for every randomized
space($S$) algorithm $A$ and every input to it we have that
$$\|\calD_{y}(A(y)) - \calD_x(A(G(x))\|_1 < \epsilon,$$
where $y$ is chosen uniformly at random in $\{0, 1\}^n$, $x$ uniformly in $\{0, 1\}^m$, and $\calD(\cdot)$ is the distribution of $\cdot$ as a vector of probabilities.
\end{definition}
Nisan's PRG~\cite{n92} is a pseudorandom generator for space $S$ with parameter $2^{-S}$ that takes a seed of length $O(S\log{R})$ bits to $R$ bits.
The total space used by Sieve and Selector is $O(\log n)$ bits for the algorithm workspace and $O(\log^2n)$ bits to store the hash functions.

We will be able to apply Nisan's PRG because Sieve only accesses the randomness in $O(\log n)$ bit chunks, where the $r$th chunk generates the 4-wise independent random variables needed for the $r$th round, namely $B_{r1},\ldots,B_{rn}$ and the bits for two instances of the AMS sketch.
We can discard the AMS sketches at the end of each round, but in order to compute its output after reading the entire stream, Selector needs access to the bit sequence $b_1,b_2,\ldots,b_R$ as well as $B_{ri}$, for $r\in[R]$ and $i\in[n]$.
Storing the $B$ random variables, by their seeds, requires $O(\log^2n)$ bits.
This poses a problem for derandomization with Nisan's PRG because it means that Sieve and Selector are effectively a $O(\log^2n)$ space algorithm, even though most of the space is only used to store random bits.

We will overcome this difficulty by derandomizing an auxiliary algorithm.
The auxiliary algorithm computes a piece of the information necessary for the outcome, specifically for a given item $j\in[n]$ in the stream the auxiliary item will compute $N_j:=\#\{r\mid b_r = B_{rj}\}$ the number of times $j$ is on the ``winning side'' and compare that value to $3R/4$.
Recall that the Selector outputs as the heavy hitter a $j$ that maximizes $N_j$.
By Lemma~\ref{lem: pair correctness} for the AMS case, $\E  N_j$ is no larger than $(\frac{1}{2}+3\delta)R$, if $j$ is not the heavy element, and $\E  N_H$ is at least $(1-3\delta)R$ if $H$ is the heavy element.
When the Sieve is implemented with fully independent rounds, Chernoff's Inequality implies that $N_H> 3R/4$ or $N_j\leq 3R/4$ with high probability.
When we replace the random bits for the independent rounds with bits generated by Nisan's PRG we find that for each $j$ with high probability $N_j$ remains on the same side of $3R/4$.

Here is a formal description of the auxiliary algorithm.
The auxiliary algorithm takes as input the sequence $q_0,q_1,\ldots,q_R$ (which is independent of the bits we want to replace with Nisan's PRG), the stream $S$, and an item label $j$, and it outputs whether $N_j>3R/4$.
It initializes $N_i=0$, and then for each round $r=1,\ldots,R$ it draws $O(\log n)$ random bits and computes the output~$b_r$ of the round.
If $b_r=B_{rj}$ then $N_i$ is incremented, and otherwise it remains unchanged during the round.
The random bits used by each round are discarded at its end.
At the end of the stream the algorithm outputs 1 if $N_j>3R/4$.

\begin{lemma}\label{lem: auxiliary for Nisan}
Let $X\in\{0,1\}$ be the bit output by the auxiliary algorithm, and let $\tX\in\{0,1\}$ be the bit output by the auxiliary algorithm when the random bits it uses are generated by Nisan's PRG with seed length $O(\log n\log\log n)$.  Then
\[ |\Pr (X=1)-\Pr (\tX=1)|\leq \frac{1}{n^2}.\]
\end{lemma}
\begin{proof}
The algorithm uses $O(\log{n})$ bits of storage and $O(\log^2n)$ bits of randomness.
The claim follows by applying Nisan's PRG~\cite{n92} with $\epsilon=1/n^2$ and seed length $O(\log n\log\log n)$.
\end{proof}

\begin{theorem}\label{thm: sieve space}
Sieve and Selector can be implemented with $O(\log(n)\log\log{n})$ random bits.
\end{theorem}
\begin{proof}
Let $N_j$ be the number of rounds $r$ for which $b_r = B_{rj}$ when the algorithm is implemented with independent rounds, and let $\tN_j$ be that number of rounds when the algorithm is implemented with Nisan's PRG.
Applying Lemma~\ref{lem: auxiliary for Nisan} we have for every item $j$ that $|\Pr (\tN_j>3R/4)-P(N_j>3R/4)|\leq 1/n^2$.
Thus, by a union bound, the probability that the heavy hitter $H$ is correctly identified changes by no more than $n/n^2=1/n$.
The random seed requires $O(\log n\log\log n)$ bits of storage, and aside from the random seeds the algorithms use $O(\log n)$ bits of storage.
Hence the total storage is $O(\log n\log\log n)$ bits.
\end{proof}

%% file: Sec-F2always.tex
One approach to tracking $F_2$ at all times is to use the median of $O(\log n)$ independent copies of an $F_2$ estimator like the AMS algorithm~\cite{ams99}.
A Chernoff bound drives the error probability to $1/\poly(n)$, which is small enough for a union bound over all times, but it requires $O(\log^2n)$ bits of storage to maintain all of the estimators. 
The Chaining Inequality allows us to get a handle on the error during an interval of times.
Our approach to tracking $F_2$ at all times is to take the median of $O(\log\frac{1}{\epsilon}+\log\log{n})$ Bernoulli processes.
In any short enough interval---where $F_2$ changes by only a $(1+\Omega(\epsilon^2))$ factor---each of the processes will maintain an accurate estimate of $F_2$ for the entire interval, with constant probability.
Since there are only $O(\epsilon^{-2}\log^2(n))$ intervals we can apply Chernoff's Inequality to guarantee the tracking on every interval, which gives us the tracking at all times.
This is a direct improvement over the $F_2$ tracking algorithm of~\cite{huang2014tracking} which for constant $\eps$ requires $O(\log n(\log n + \log\log m))$ bits.

The algorithm has the same structure as the AMS algorithm, except we replace their sketches with instances of \textsc{JLBP}.

\begin{algorithm}
  \begin{algorithmic}
    \Procedure{F2Always}{Stream $S$}
    \State $N\gets O(\frac{1}{\epsilon^2})$, $R\gets O(\log(\frac{1}{\epsilon^2}\log n))$
    \State $X_{i,r}^{(t)}\gets\textsc{JLBP}(S)$ for $i\in[N]$ and $r\in[R]$.\Comment{Use a $(1\pm\frac{\epsilon}{3})$-embedding $T$ in this step.}
    \State $Y^{(t)}_r = \frac{1}{N}\sum_{i=1}^N \|X_{i,r}^{(t)}\|_2^2$
    \State \Return {$\hat{F}_2^{(t)}=\text{median}_{r\in R}\{Y_r^{(t)}\}$ at each time $t$}
    \EndProcedure
  \end{algorithmic}
  \caption{An algorithm for approximating $F_2$ at all points in the stream.}\label{algo: F2Always}
\end{algorithm}

\begin{lemma}\label{lem: F2always one piece}
Let $N=O(\frac{1}{\delta\epsilon^2})$ and let $X_{i}^{(t)}$, for $i=1,\ldots,N$, be independent copies of the output of $\textsc{JLBP}(S)$ using a fixed $(1\pm\frac{\epsilon}{8})$-embedding $T$ on an insertion only stream $S$. 
Let $Y_t = \frac{1}{N}\sum_{i=1}^N\|X_i^{(t)}\|_2^2$.
Suppose that for two given times $1\leq u<v\leq m$ the stream satisfies $256C^2F_2^{(u:v)}\leq \epsilon^2 F_2^{(u)}$, where $F_2^{(u:v)}=\sum_{i=1}^n(f_i^{(u:v)})^2$ is the second moment of the change in the stream.
Then
\[\Pr\left(|Y_t-F_2^{(t)}| \leq \epsilon F_2^{(t)}\text{, for all }u\leq t\leq v\right)\geq 1-2\delta.\] 
\end{lemma}
\begin{proof}
We first write $|Y_t-F_{2}^{(t)}| \leq |Y_t - Y_u| + |Y_u-F_{2}^{(u)}| + |F_2^{(t)}-F_2^{(u)}|$.
It follows from the arguments of AMS and the fact that $T$ is a $(1\pm\epsilon/8)$-embedding that, with an appropriate choice for $N=O(\frac{1}{\delta\epsilon^2})$, we arrive at \begin{equation}\label{eq: AMS for F2always}
\Pr(|Y_u-F_2^{(u)}|\leq \frac{\epsilon}{4} F_2^{(u)})\geq 1-\delta.
\end{equation}

For the third term we have $F_2^{(t)}\geq F_2^{(u)}$ because $t\geq u$ and the stream is insertion only.
We can bound the difference with
\begin{align*}
F_2^{(t)} 
= \| f^{(u)} + f^{(u:t)}\|_2^2 
\leq \|f^{(u)}\|_2^2\left(1 + \frac{\|f^{(u:t)}\|_2}{\|f^{u}\|_2}\right)^2 
\leq F_2^{(u)}(1+\frac{\epsilon}{4}),
\end{align*}
where the last inequality follows because $C\geq 2$ and $\epsilon\leq 1/2$.

For the first term, since $X_i^{(t)}$, $i\in[n]$, are independent $d$-dimensional Bernoulli processs, it follows that
\[X^{(t)} = \frac{1}{\sqrt{N}}((X_{1}^{(t)})^T,(X_{2}^{(t)})^T,\ldots,(X_{N}^{(t)})^T)^T\] 
is an $Nd$-dimensional Bernoulli process.
By Lemma~\ref{lem: JLBP is tame} and due to the fact that $X^{(t)}$ can be represented as an output of JLBP procedure, the process $X^{(u:t)} = X^{(t)}-X^{(u)}$, is a tame process, so with probability at least $1-\delta$, for all $u\leq t\leq v$ we have
\[\|X^{(u:t)}\|_2^2 \leq C^2 \sum_{j=1}^n (f_{j}^{(u:v)})^2.\]
Therefore, assuming the inequality inside~\eqref{eq: AMS for F2always},
\begin{align*}
Y_t = \|X^{(u)} + X^{(u:t)}\|_2^2
 &\leq Y_u\left(1 + \frac{\|X^{(u:t)}\|_2}{\|X^{(u)}\|_2}\right)^2
 \leq Y_u\left(1 + \frac{\sqrt{1+\eps}}{\sqrt{1-\eps}}\frac{\|F^{(u:t)}\|_2}{\|F^{(u)}\|_2}\right)^2\\
&\leq Y_u\left(1+\frac{2\epsilon}{16C}\right)^2
\leq F_2^{(u)}(1+\epsilon/4),
\end{align*}
where the last inequality follows because $C\geq 2$ and $\epsilon\leq 1/2$.
The reverse bound $Y_t\geq F_2^{(u)}(1-\epsilon/4)$ follows similarly upon applying the reverse triangle inequality in place of the triangle inequality.

With probability at least $1-2\delta$,
\[|Y_t-F_{2}^{(t)}| \leq |Y_t - Y_u| + |Y_u-F_{2}^{(u)}| + |F_2^{(t)}-F_2^{(u)}|\leq \epsilon F_2^{(u)}\leq \epsilon F_2^{(t)}.\]
\end{proof}

\begin{theorem}
Let $S$ be an insertion only stream and, for $t=1,2,\ldots,m$, let $F_2^{(t)} = \sum_{i=1}^n (f_{i}^{(t)})^2$ and let $\hat{F}_2^{(t)}$ be the value that is output by Algorithm~\ref{algo: F2Always}.
Then
\[P(|\hat{F}_2^{(t)}-F_2^{(t)}|\leq \epsilon F_2^{(t)}\text{, for all }t\in[m])\geq 2/3.\]
The algorithm uses $O\left(\frac{1}{\epsilon^2}\log n \left(\log\log n + \log \frac{1}{\eps}\right)\right)$ bits of space.
\end{theorem}
\begin{proof}
By Lemma~\ref{thm: meka JL}, the (single) matrix used by all instances of JLBP is a $(1\pm\epsilon/3)$-embedding with probability at least $0.99$, henceforth assume it is so.
Let $q_0=0$ and 
\[q_i = \max_t\left\{t\;|F_2^{(t)} \leq (1+\frac{\eps^2}{256C^2})^i\right\},\]
until $q_K = m$ for some $K$.
Notice that $K=O(\frac{1}{\epsilon^2}\log n)$.
Here, $C$ is the constant from Definition~\ref{def: tame}.

By definition of $q_i$ and using the fact that $(a-b)^2 \leq a^2 - b^2$ for real numbers $0\le b\le a$ we have $F_2^{(q_i:q_{i+1})}\leq (F_2^{(q_{i+1})}-F_2^{(q_i)})\leq \frac{\eps^2}{256C^2}F_2^{(q_i)}$.

Applying Lemma~\ref{lem: F2always one piece} with $\delta=1/10$, we have, for every $r\in[R]$ and $i\geq 0$ that
\[P(|Y_r^{(t)} - F_2^{(t)}| \leq \epsilon F_2^{(t)},\text{ for all }q_i\leq t\leq q_{i+1})\geq 0.8.\]
Thus, by Chernoff bound, the median satisfies
\[P(|\hat{F}_2^{(t)} - F_2^{(t)}| \leq \epsilon F_2^{(t)},\text{ for all }q_i\leq t\leq q_{i+1})\geq 1- e^{-R/12} \geq 1-\frac{1}{4K},\]
by our choice of $R = 12\log{4K} = O(\log(\epsilon^{-2}\log{n}))$.
Thus, by a union bound over all of the intervals and the embedding $T$ we get
\[P(|\hat{F}_2^{(t)} - F_2^{(t)}| \leq \epsilon F_2^{(t)},\text{ for all }t\in[m])\geq \frac{2}{3},\]
which completes the proof of correctness.

The algorithm requires, for the matrix $T$, the JL transform of Kane, Meka, and Nelson~\cite{kane2011almost} with a seed length of $O(\log(n)\log(\frac{1}{\epsilon}\log{n}))$ bits, and it takes only $O(\log(n/\epsilon))$ bits of space to compute any entry of $T$.
The algorithm maintains $NR=O(\epsilon^{-2}\log(\frac{1}{\epsilon}\log{n}))$ instances of \textsc{JLBP} which each requires $O(\log n)$ bits of storage for the sketch and random bits.
Therefore, the total storage used by the algorithm is $O(\epsilon^{-2}\log(n)\log(\frac{1}{\epsilon}\log{n}))$.
\end{proof}
This immediately implies Theorem~\ref{thm: f2always}.